\documentclass[acmsmall,screen]{acmart}
\usepackage{booktabs, multirow}
\usepackage{bm}
\usepackage{algpseudocode}
\usepackage{xcolor}
\usepackage{subcaption}
\usepackage{algorithm}

\usepackage{amsthm}
\theoremstyle{remark}
\newtheorem{remark}{Remark}

\renewcommand\footnotetextcopyrightpermission[1]{}
\setcopyright{none}
\acmConference{}
\acmBooktitle{}
\acmYear{}
\acmDOI{}
\acmISBN{}
\acmPrice{}

\begin{document}

% \title{Privacy-Preserving Federated Item Cold-Start Recommendation via Diffusion-Based Generation}
\title{MDiffFR: Modality-Guided Diffusion Generation for Cold-start Items in Federated Recommendation}


\author{Kang Fu}
\authornote{Both authors contributed equally to this research.}
\email{kangfu@bjtu.edu.cn}
\affiliation{%
  \institution{Key Laboratory of Big Data \& Artificial Intelligence in Transportation, Ministry of Education}
  \city{Beijing}
  % \state{Beijing}
  \country{China}
}
\affiliation{%
  \institution{School of Computer Science and Technology, Beijing Jiaotong University}
  \city{Beijing}
  % \state{Beijing}
  \country{China}
}

\author{Honglei Zhang}
\authornotemark[1]
\email{honglei.zhang@bjtu.edu.cn}
\affiliation{%
  \institution{Key Laboratory of Big Data \& Artificial Intelligence in Transportation, Ministry of Education}
  \city{Beijing}
  % \state{Beijing}
  \country{China}
}
\affiliation{%
  \institution{School of Computer Science and Technology, Beijing Jiaotong University}
  \city{Beijing}
  % \state{Beijing}
  \country{China}
}

\author{Xuechao Zou}
\email{xuechaozou@bjtu.edu.cn}
\affiliation{%
  \institution{School of Computer Science and Technology, Beijing Jiaotong University}
  \city{Beijing}
  % \state{Beijing}
  \country{China}
}

% \author{Yidong Li}
\author{Yidong Li}
\authornote{Corresponding author. Email: ydli@bjtu.edu.cn}
\email{ydli@bjtu.edu.cn}
\affiliation{%
  \institution{Key Laboratory of Big Data \& Artificial Intelligence in Transportation, Ministry of Education}
  \city{Beijing}
  % \state{Beijing}
  \country{China}
}
\affiliation{%
  \institution{School of Computer Science and Technology, Beijing Jiaotong University}
  \city{Beijing}
  % \state{Beijing}
  \country{China}
}

% \affiliation{%

% \affiliation{%

% \affiliation{%

% \affiliation{%

% \affiliation{%

\renewcommand{\shortauthors}{Kang Fu et al.}

\begin{abstract}
Federated recommendations (FRs) provide personalized services while preserving user privacy by keeping user data on local clients, which has attracted significant attention in recent years. However, due to the strict privacy constraints inherent in FRs, access to user-item interaction data and user profiles across clients is highly restricted, making it difficult to learn globally effective representations for new (cold-start) items. Consequently, the item cold-start problem becomes even more challenging in FRs. Existing solutions typically predict embeddings for new items through the attribute-to-embedding mapping paradigm, which establishes a fixed one-to-one correspondence between item attributes and their embeddings. However, this one-to-one mapping paradigm often fails to model varying data distributions and tends to cause embedding misalignment, as verified by our empirical studies. To this end, we propose MDiffFR, a novel generation-based modality-guided diffusion method for cold-start items in FRs. In this framework, we employ a tailored diffusion model on the server to generate embeddings for new items, which are then distributed to clients for cold-start inference. To align item semantics, we deploy a pre-trained modality encoder to extract modality features as conditional signals to guide the reverse denoising process. Furthermore, our theoretical analysis verifies that the proposed method achieves stronger privacy guarantees compared to existing mapping-based approaches. Extensive experiments on four real datasets demonstrate that our method consistently outperforms all baselines in FRs. 
\end{abstract}

\begin{CCSXML}
<ccs2012>
   <concept>
       <concept_id>10002951.10003227.10003351.10003269</concept_id>
       <concept_desc>Information systems~Collaborative filtering</concept_desc>
       <concept_significance>500</concept_significance>
       </concept>
   <concept>
       <concept_id>10002978.10003029.10011150</concept_id>
       <concept_desc>Security and privacy~Privacy protections</concept_desc>
       <concept_significance>500</concept_significance>
       </concept>
 </ccs2012>
\end{CCSXML}

\ccsdesc[500]{Information systems~Collaborative filtering}
\ccsdesc[500]{Security and privacy~Privacy protections}

\keywords{Federated Recommendation, Item Cold-Start, Diffusion Model, Privacy Guarantee}


\maketitle

\section{Introduction}
Recommender Systems (RSs) aim to identify items of potential interest to users from massive data collections, which require accurately modeling user preferences and item characteristics based on historical interaction data. However, this becomes particularly challenging when new (cold-start) items arrive without any historical interactions.
% To improve the performance of RS, researchers have characterized the emergence of new items without interaction history as the item cold-start problem, recognizing it as a key factor that can hinder accurate recommendations.
Item cold-start problem is a persistent and widely recognized challenge in RSs, frequently encountered in practical applications across various domains \cite{JIIS_survey}, such as e-commerce platforms \cite{e-commerce}, social recommendation scenarios \cite{social}, and online content services \cite{online}. In centralized settings, recommender models typically leverage item modality data, along with the interaction history of items sharing similar modal characteristics, as auxiliary signals to compensate for the lack of interaction history for cold-start items, effectively alleviating this challenge to a significant extent \cite{CLCRec, CGRC}. However, centralized approaches require collecting all user data to the server for model training, which raises privacy concerns and potential risks of data leakage, especially after the implementation of regulations like the General Data Protection Regulation \cite{GDPR}. To address this challenge, Federated Recommendations (FRs) have emerged as a promising distributed collaborative training architecture \cite{FCF, FL}.

\begin{figure}[h]
\centering
\includegraphics[width=\columnwidth]{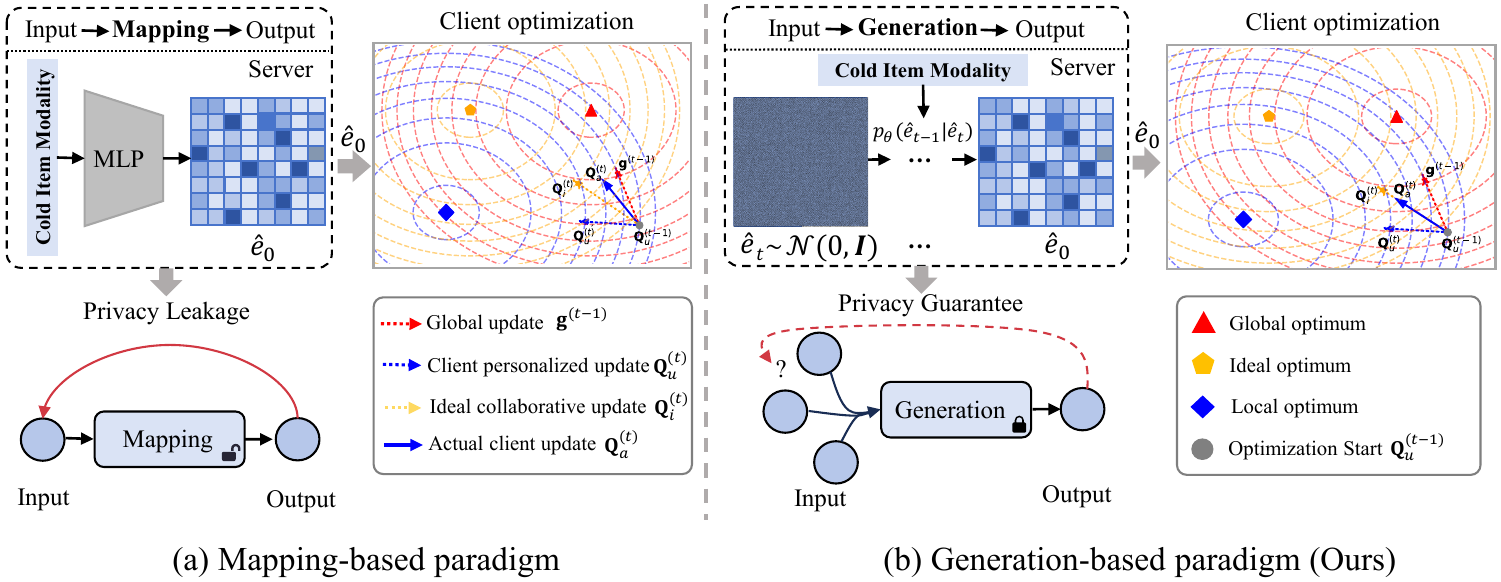}
\caption{Comparison between our generation-based and conventional mapping-based cold-start methods in federated settings. $\hat{e}_0$ and $\hat{e}_t$ denote predicted embeddings and noise samples, respectively. $p_\theta(\hat{e}_{t-1}|\hat{e}_t)$ represents the reverse denoising for reconstructing the data. Our method expands the solution space against inversion attacks, achieving enhanced privacy guarantees. On the client side, it ensures that the local optimization direction remains consistent with the ideal optimization trajectory.}
\label{fig:compare_framework}
\end{figure}

Specifically, FRs enable multiple clients to jointly train a global model without sharing their raw interaction data. Each client updates its local model using private data on the device side and uploads model parameters to a central server, where a global model is aggregated. The distributed learning framework ensures that raw user data never leaves local devices during training, effectively preserving user privacy \cite{FedAvg, LightFR}. Nevertheless, this privacy constraint makes the cold-start issue more challenging in FRs, since the server is unable to access all historical interactions to analyze the patterns in similar items, while the inherent sparsity of client-side data also limits the model’s ability to learn effective representations. IFedRec \cite{IFedRec} is an early attempt to address the item cold-start problem in FRs by deploying a multilayer perceptron (MLP) on the server to map item textual modalities (attributes) to item embeddings, as shown in Fig. \ref{fig:compare_framework}(a). The mapping function achieved by MLP is trained on the server using the global warm item embeddings as target representations to learn the relationship from item modalities to embeddings. Meanwhile, a regularization term is introduced on the client side to enforce the locally item embeddings to conform to the embeddings predicted by the server-side mapping function and distributed to clients, thereby ensuring that client-side models remain compatible with the embeddings of cold-start items. This design effectively alleviates the item cold-start problem under federated settings.

% \begin{figure}[t]
% \centering
% \includegraphics[width=\columnwidth]{./image/compare_framework_gradient.pdf}
% \caption{Comparison between our generation-based and conventional mapping-based cold-start methods in federated settings. $\hat{e}_0$ and $\hat{e}_t$ denote predicted embeddings and noise samples, respectively. $p_\theta(\hat{e}_{t-1}|\hat{e}_t)$ represents the reverse denoising for reconstructing the data. Our method expands the solution space against inversion attacks, achieving enhanced privacy guarantees. On the client side, it ensures that the client’s optimization direction is not perturbed by the global cold-start optimum.}
% \label{fig:compare_framework}
% \end{figure}

However, the mapping-based paradigm, which employs an MLP to map item modalities to their corresponding embeddings \cite{IFedRec}, has \textbf{three main limitations}: \textbf{(1)} the deterministic nature of the mapping makes the output highly sensitive to the input, exposing it to significant risks from inversion attacks. As illustrated in Fig. \ref{fig:compare_framework}(a), because the mapping is deterministic, an adversary can train a shadow model to infer the original input from the output, which poses a significant risk of privacy leakage; \textbf{(2)} in federated settings, IFedRec introduces a regularization term on the client side to enforce the locally learned item embeddings to conform to the distributed item embeddings predicted by the MLP on the server during the training phase, which ensures that the locally learned item embeddings remain compatible with the global predicted embedding space. However, this regularization term introduces a non-local objective into the loss function, causing the parameter update direction to deviate from the original optimization trajectory driven by each client’s local data and global collaborative information \cite{NoRegular, FedDC}. As shown in Fig. \ref{fig:compare_framework}(a), the local optimization direction $\bm{Q_a^{(t)}}$ should be determined jointly by the local data optimization $\bm{Q_u^{(t)}}$ and the global collaborative direction $\bm{g^{(t-1)}}$. Yet, the regularization term forces the local updates to lean excessively toward the global optimal direction, thereby severely undermining local personalized information and disrupting the model’s inherent optimization trajectory; \textbf{(3)} this approach merely learns a one-to-one static mapping function from modality features to item embeddings, and thus fails to capture the underlying data distribution. We conduct a preliminary experiment to examine the distribution of embeddings for items in the Food dataset \cite{NineRec} (experimental settings are provided in Section 6). We visualize the distributions of warm and cold-start item embeddings obtained by the mapping-based and generation-based methods in the final training round. As shown in Fig. \ref{fig:embedding_comparison}(a), the mapping-based method leads to embedding misalignment, resulting in the generated embeddings for cold-start items to cluster around limited regions and fail to capture the diversity of the embedding space.

% In addition, to ensure that the model's output remains consistent with warm item embeddings in FR, IFedRec introduces a regularization term during client training to align local embeddings with cold-start item embeddings, However, this regularization term introduces a non-local objective into the loss function, which causes the parameter update direction to deviate from the original optimization objective driven by local data \cite{NoRegular}.

\begin{figure}[t]
    \centering
    % \subfloat[Mapping-based embeddings]{
    \subfloat[Mapping-based method]{
        \includegraphics[width=0.45\columnwidth]{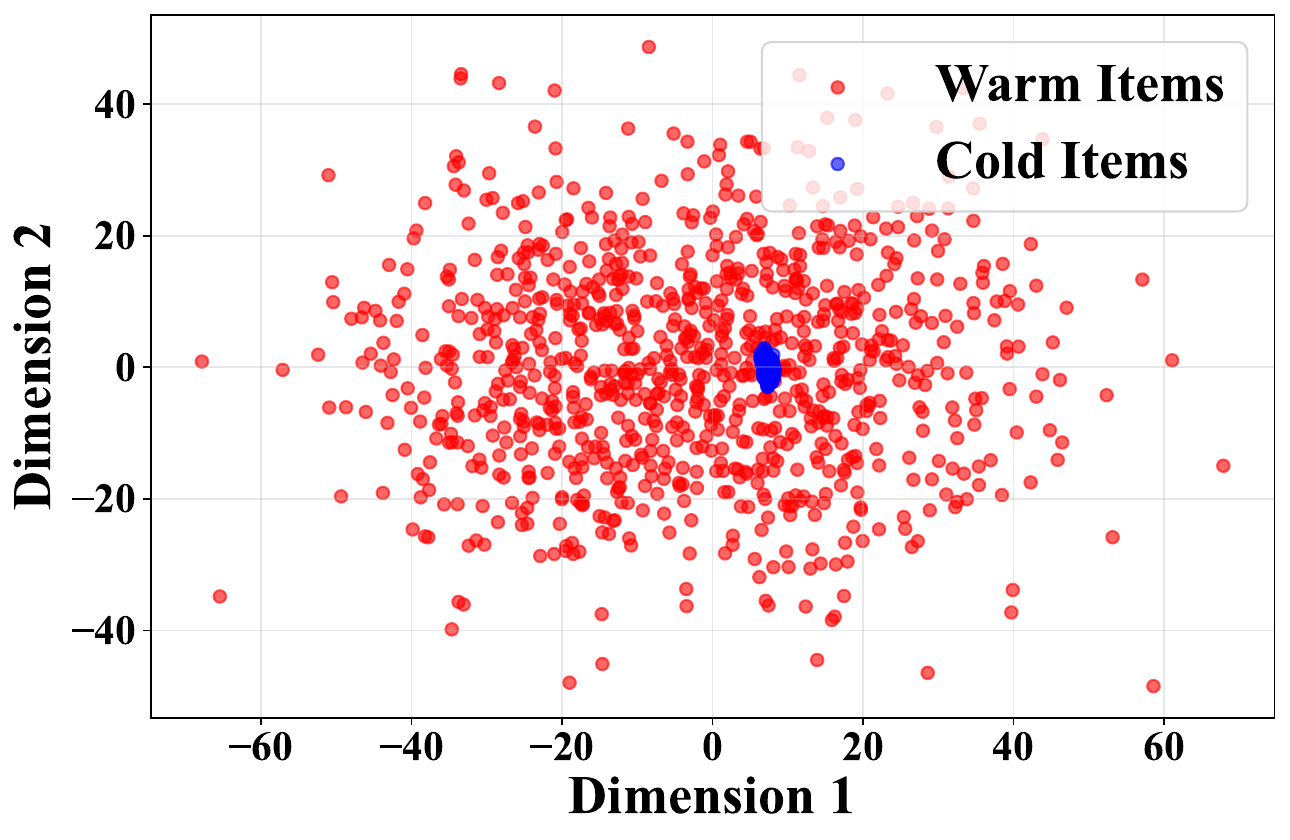}
        \label{fig:mlp_emb}
    }
    \hfill
    % \subfloat[Generation-based embeddings]{
    \subfloat[Generation-based method]{
        \includegraphics[width=0.45\columnwidth]{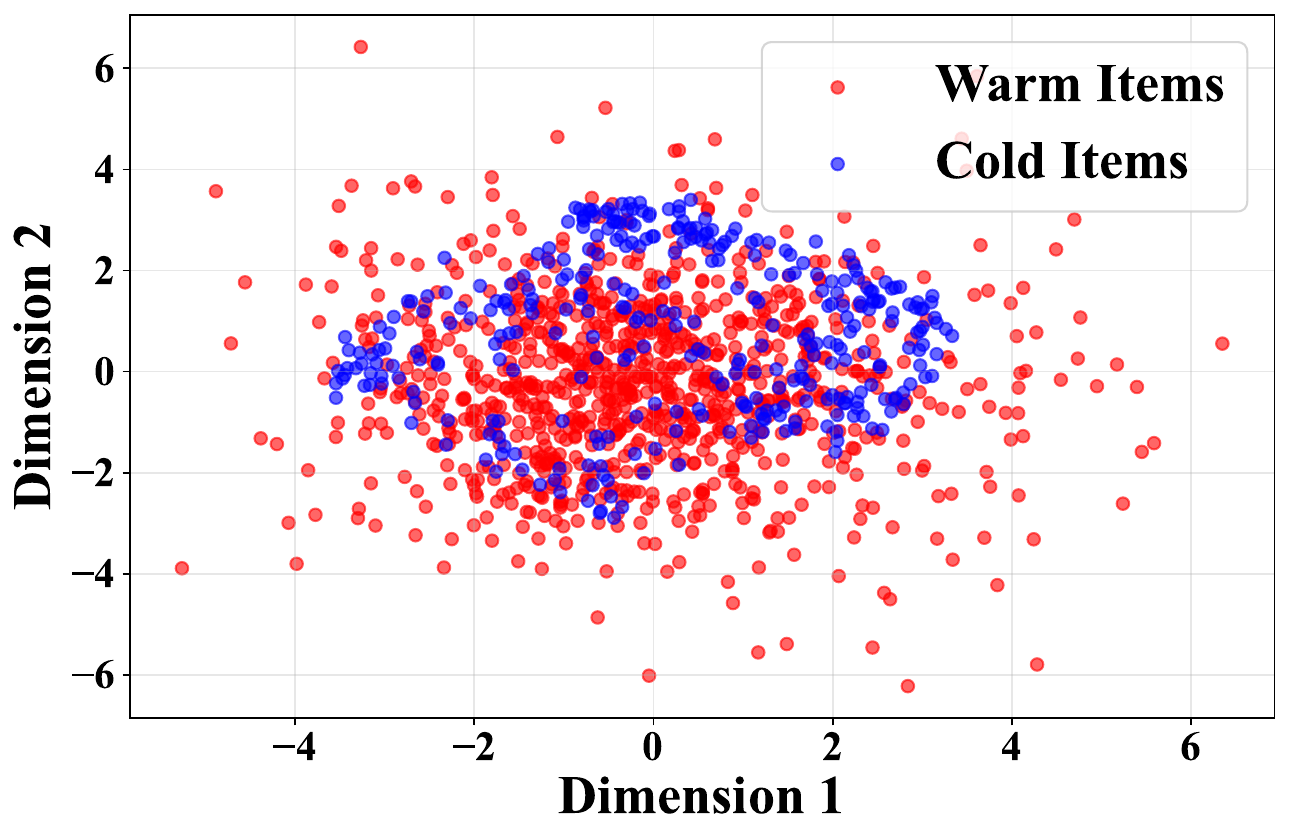}
        \label{fig:diff_emb}
    }
    \caption{Comparison of the embedding distributions generated by the Mapping-based and Generation-based methods. $\textbf{Warm items}$ denote those with user interactions, while $\textbf{Cold items}$ have no user interactions.}
    \label{fig:embedding_comparison}
\end{figure}

% Diffusion models, as a prominent class of generative methods, learn the underlying probability distribution of the training data to synthesize new samples \cite{DDPM}. They gradually add Gaussian noise in the forward process and reconstruct the data in the reverse process by iterative denoising. Such models have achieved remarkable success in image synthesis \cite{LiDAR-4D, selfguidance}. In contrast to the image domain, which focuses on reconstructing visual content at the pixel level, the FR primarily aims to reconstruct the distribution of item embeddings, which captures collaborative information about user preferences. \cite{FedRAP, embedding2preference}. Drawing on this analogy, we employ a tailored generation-based diffusion model to iteratively fit the data distribution of item embeddings, aiming to alleviate the embedding misalignment issue and ultimately generate better embeddings for cold-start items. 

Motivated by these limitations, we consider leveraging a generation-based diffusion model to tackle the item cold-start problem. Unlike mapping-based methods, diffusion models can capture the underlying data distribution\cite{DR_Survey}, thereby enabling a more effective modeling of user preferences.
Existing diffusion-based recommendation methods in centralized settings typically directly generate items by modeling users' historical interaction distributions to learn their preferences \cite{DiffRec}.
% , while others learn the distribution of item embeddings to better capture the underlying user preferences \cite{DiffuRec, Diff-MSR}. 
However, in federated settings, the inherent sparsity and decentralized nature of client-side data make it infeasible for diffusion models to learn reliable preference distributions solely from local interactions. 
Consequently, directly deploying diffusion models on the client side to generate cold-start items is impractical in FRs.

Based on the above analysis, we propose a \textbf{M}odality-guided \textbf{Diff}usion-based \textbf{F}ederated \textbf{R}ecomm\-endation model (\textbf{MDiffFR}) for the item cold-start problem in FRs, as illustrated in Fig. \ref{fig:compare_framework}(b). 
Our method deploys the model on the server to learn the distribution of global item embeddings, which implicitly reflects the users’ preference distribution\cite{embedding2preference}. 
During the training phase, modality features of items are incorporated as conditional guidance, enabling the diffusion model to not only capture the global preference distribution but also align the generated embeddings with the semantic features of items. During the inference phase, we sample noise from a prior distribution and progressively generate embeddings for cold-start items via the reverse denoising process. The generated item embeddings are then distributed to each client for personalized recommendations.

Compared to traditional mapping-based approaches, the MDiffFR possesses \textbf{three key advantages}: 
\textbf{(1)} it introduces stochasticity into the generation process by sampling noise at each step of the reverse denoising, thereby providing stronger privacy guarantees. As shown in Fig. \ref{fig:compare_framework}(b), the inherent stochasticity in the generation process leads to an extremely large solution space for any inverse mapping, making it practically infeasible for an adversary to recover the original input from the model’s output;
\textbf{(2)} since the diffusion-based generation process is trained solely on the globally aggregated item embeddings to ensure that the learned distribution approximates the distribution of item embeddings on the server, it does not impose any modifications or constraints on client-side training. As a result, the optimization of item embeddings on each client remains fully independent and is not influenced by the globally learned distribution. As shown in Fig. \ref{fig:compare_framework}(b), the local optimization direction $\bm{Q_a^{(t)}}$ is determined solely by the locally interactive data $\bm{Q_u^{(t-1)}}$ and the global collaborative information $\bm{g^{(t-1)}}$, thereby preserving the ideal local optimization direction. In contrast to Fig. \ref{fig:compare_framework}(a), where the optimization direction of the local item embeddings is forced to bias toward the distributed item predicted by the MLP on the server, our method enables the learned distribution to progressively approach the global item embeddings, thereby preserving the inherent optimization directions of the local items;
\textbf{(3)} it can capture the underlying data distribution of item embeddings by iteratively reconstructing the original data from noise, rather than merely learning a simple deterministic mapping. As illustrated in Fig. \ref{fig:embedding_comparison}(b), our generation-based method captures a more meaningful embedding distribution compared to the mapping-based method, generating embeddings that better align with the true distribution. Our main contributions are as follows:
% In addition, solely relying on the diffusion model to fit the data distribution and generate embeddings overlooks the influence of item-specific modality. Therefore, we introduce item modality as a condition to guide the diffusion process, to ensure that the generated item embeddings align with item semantics.

\begin{itemize}
    \item We propose MDiffFR, the first generation-based framework for cold-start items in FRs, leveraging iterative denoising to model the underlying distribution of item embeddings while preserving the optimization direction of client-side models.
    \item To ensure semantic alignment of the generated item embeddings, we introduce modality features as conditional guidance for the model’s reverse generation process.
    \item We discuss the inherent privacy limitations of existing mapping-based methods and theoretically analyze the strong privacy protection capability of MDiffFR.
    \item Extensive experiments on four real-world datasets show that our method achieves superior performance for item cold-start recommendation and stronger privacy guarantees. 
\end{itemize}

The rest of the paper is structured as follows. Section 2 reviews the existing studies related to federated recommendation and clarifies the position of our work within this research landscape. Section 3 introduces the preliminary knowledge required to understand MDiffFR. Section 4 presents the overall framework of MDiffFR and provides detailed descriptions of each component. Section 5 analyzes the privacy guarantees of MDiffFR under the inversion attack. Section 6 reports experimental results on four real-world datasets to validate the effectiveness of MDiffFR in terms of both recommendation performance and privacy protection. Finally, Section 7 concludes this paper.

\section{Related work}
In this section, we will briefly review the related work and outline the position of our work within the research landscape.
\subsection{Cold-start Recommendation}
Cold-start recommendation is prevalent in real-world applications and can be broadly categorized into item cold-start \cite{CB1-ItemCS, CB2-ItemCS, CF1-ItemCS, ItemCS, ALDI}, user cold-start \cite{UserCS, userCS3, userCS2}, user-item cold-start \cite{UserItemCS}, and system-level cold-start scenarios \cite{PromptRec}. Item cold-start is a critical component of the cold-start problem and significantly impacts the overall performance of recommendation systems. In centralized recommendation systems, item cold-start is primarily addressed through collaborative filtering and content-based recommendation \cite{shi2014collaborative,JIIS_survey}. Collaborative filtering methods construct similarities between users to facilitate recommendations for new items \cite{guan2024hybrid}. Content-based approaches describe cold-start items using their modality and establish similarity relationships with existing items \cite{IFedRec,CGRC}. However, in such centralized recommendation scenarios, all user information must be uploaded to the server, including sensitive data such as interaction histories, which poses serious privacy risks \cite{FedAvg, FCF}.

\subsection{Federated Recommendation}
FRs train models locally on client devices and aggregate the updates on a central server, ensuring that users' private data never leaves their local devices and thereby effectively preserving user privacy \cite{FedAvg, FCF, FedCA, Debiased, User-Governed}. 
FedMF \cite{FedMF} employs a matrix factorization approach under the federated setting to enable secure recommendations.
FedNCF \cite{FedNCF} introduces neural networks as the scoring function to capture nonlinear collaborative relationships.
FedPA \cite{FedPA} adopts a dual-tower architecture consisting of user-level and user-group-level models on the client side to achieve more personalized recommendation performance.
FedRAP \cite{FedRAP} utilizes additive personalization embeddings to preserve client-specific information and employs sparsity-constrained global item embeddings to retain non-personalized information, thereby improving both personalization and communication efficiency.
HeteFedRec \cite{HeteFedRec} proposes a heterogeneous federated recommendation framework that supports personalized model size assignment for participants. It further introduces a heterogeneous aggregation strategy to effectively integrate recommendation models of varying sizes, thereby improving overall recommendation performance. 
Despite these advances, existing FRs primarily focus on personalized recommendation \cite{PFedRec, FedPA, FedRAP}, efficient model aggregation \cite{Fedfast}, communication efficiency \cite{RFRec}, and model heterogeneity \cite{HeteFedRec}, while the cold-start problem has been relatively underexplored. 

New items without interaction history can substantially degrade overall recommendation performance, making it crucial to address this problem.
IFedRec \cite{IFedRec} is the first to address the item cold-start problem in FRs by deploying an attribute network on the server to predict embeddings for new items. Subsequently, FR-CSU \cite{FR-CSU} extends this line of research by investigating the user cold-start problem. In this paper, we further explore the item cold-start problem in FRs by addressing two critical issues in existing methods: embedding misalignment and privacy leakage.

\subsection{Diffusion-based Recommendation}

Diffusion models have achieved remarkable success in various domains such as image synthesis and text generation, owing to their strong ability to capture complex data distributions and generate high-quality samples \cite{DDPM, DDIM}. Motivated by this, DiffRec \cite{DiffRec} first introduced a diffusion-based recommendation paradigm. Specifically, it leverages an encoder and the sequential information of user–item interactions to derive two variants, L-DiffRec and T-DiffRec, which further enhance the model’s performance. To address the limitation of existing sequential recommendation methods that represent items with fixed vectors, DiffuRec \cite{DiffuRec} pioneers the use of diffusion models in sequential recommendation. By modeling item representations as probability distributions, it achieves remarkable recommendation performance. In contrast to DiffuRec, DiffuASR \cite{DiffuASR} explores diffusion models from the perspective of sequence augmentation, generating pseudo interaction sequences to improve recommendation performance.
Furthermore, MCDRec \cite{MCDRec} employs diffusion models to model and fuse multi-modal information, proposing a multi-conditioned representation diffusion module and a diffusion-guided graph denoising module to integrate multi-modal features into item representations while filtering noise from user interaction histories. DiffMM \cite{DiffMM} adopts a modality-aware graph diffusion framework to achieve more effective alignment between multimodal features and collaborative information.
RecDiff \cite{RecDiff} proposes a diffusion-based social recommendation model that exploits the denoising capability of diffusion models to effectively mitigate noise contamination in social recommendation.

Moreover, diffusion models have also been utilized in centralized settings to tackle the cold-start problem in multi-scenario recommendation and click-through rate prediction tasks, demonstrating strong effectiveness \cite{Diff-MSR, CSDM}.    
In our work, we leverage the uncertainty generation and feature distribution capturing capabilities of diffusion models to explore the privacy-preserving cold-start recommendation, and mitigate the embedding misalignment between warm and cold items as well as between the server and clients under the federated setting.

\section{Preliminary}
In this section, we introduce the workflow and architecture of FRs, and subsequently present the research problem of federated item cold-start recommendation explored in this work.
\subsection{Federated Recommendation}
Federated Recommendation is a distributed collaborative learning framework consisting of a central server and multiple clients. The clients are responsible for training local models using their own data, while the server coordinates the aggregation of these local models to enable collaborative learning across clients.  
Typically, the local model on client $u$ includes three main components: the user embedding $\mathbf{e}_u$, the item embedding $\mathbf{e}_i$, and a scoring function $f(\cdot)$.  
Among them, the user embeddings $\mathbf{e}_u$ and the scoring function $f(\cdot)$ are stored locally on the client side since they involve sensitive user information, whereas the item embeddings $\mathbf{e}_i$ are uploaded to the central server for aggregation, enabling global model optimization and collaborative training among multiple clients.

Let $\mathcal{U}$ and $\mathcal{I}$ denote the sets of users and items, respectively. Each user in $\mathcal{U}$ corresponds to an individual client. On each client, the local model is trained using its own historical interaction data. Assume that the historical interaction data of client 
$u$ is denoted as $\mathcal{D}_u = \{ i_1, i_2, \ldots, i_m \}$. where each $i \in \mathcal{I}$ represents an item that user $u$ has interacted with. For each interacted item $i$, we define $y_i = 1$, indicating the existence of an interaction between the user $u$ and the item $i$. 
Accordingly, the local optimization objective of client $u$ can be formulated as:  
\begin{equation}
    \mathcal{L}_u = - \sum_{i \in \mathcal{D}_u} \left( y_i \log \hat{y}_{ui} + (1 - y_i) \log (1 - \hat{y}_{ui}) \right),
\end{equation}
where $\hat{y}_{ui} = f(\mathbf{e}_u, \mathbf{e}_i)$ denotes the predicted interaction score computed by the scoring function $f(\cdot)$ based on the user embedding $\mathbf{e}_u$ and item embedding $\mathbf{e}_i$.  

Afterwards, the clients upload the updated item embedding parameters or gradients to the central server, which performs model aggregation and redistributes the aggregated model back to the clients for the next training and recommendation round.
Throughout this process, private data such as user historical interaction data and profiles remain on the client side, ensuring privacy preservation, while the server-side aggregation facilitates the sharing of common parameters and collaborative learning among clients.

\subsection{Federated Item Cold-start Recommendation}
In federated item cold-start recommendation, the user set $\mathcal{U}$ remains fixed, while new items are continuously added to the item set $\mathcal{I}$. Since these new items have no interaction history with any user in $\mathcal{U}$, traditional federated recommendation models fail to accurately represent them in the embedding space, thereby limiting their recommendation performance.
Given a cold-start item $i_{cold}$, the predicted rating score of user $u$ for this item can be formulated as:
\begin{equation}
    \hat{y}_{i_{cold}}=f(e_u,e_{i_{cold}}),
    \label{eq:score_function}
\end{equation}
where $e_{i_{cold}}$ represents the embedding of new item $i_{cold}$. The quality of $e_{i_{cold}}$ directly determines the accuracy of the predicted score.
Therefore, federated item cold-start recommendation aims to enhance the model’s capability to handle new items by effectively generating and aligning representations for new items without compromising user privacy. 

\section{Methodology}
In this section, we provide a detailed description of the proposed method and summarize its workflow through pseudocode for clearer exposition.
\subsection{Framework Overview}
Fig. \ref{fig:framework} presents the overall framework of the Modality-guided Diffusion-based Federated Recommendation (MDiffFR), which comprises two phases: training and inference. 
During the training phase, each client trains its local model and uploads the updated item embeddings to the server. The server aggregates these embeddings as global item embeddings, which are then redistributed to all clients for the next training round. Meanwhile, the diffusion model on the server learns the distribution of global item embeddings by performing a forward process to gradually add noise and a reverse process to iteratively denoise. To align the generated embeddings with the semantic information of items, we employ a modality encoder on the server to encode item modalities, enabling the modality features to serve as conditional guidance during the reverse denoising process.
During the inference phase, the model initially samples noise from a prior distribution and extracts modality features for cold-start items through the modality encoder. These modality features serve as conditions to guide the diffusion model to progressively reconstruct the embeddings for cold-start items from the noise. The generated embeddings are then distributed to clients to facilitate cold-start recommendations.

\begin{figure*}[t]
\centering
\includegraphics[width=\textwidth]{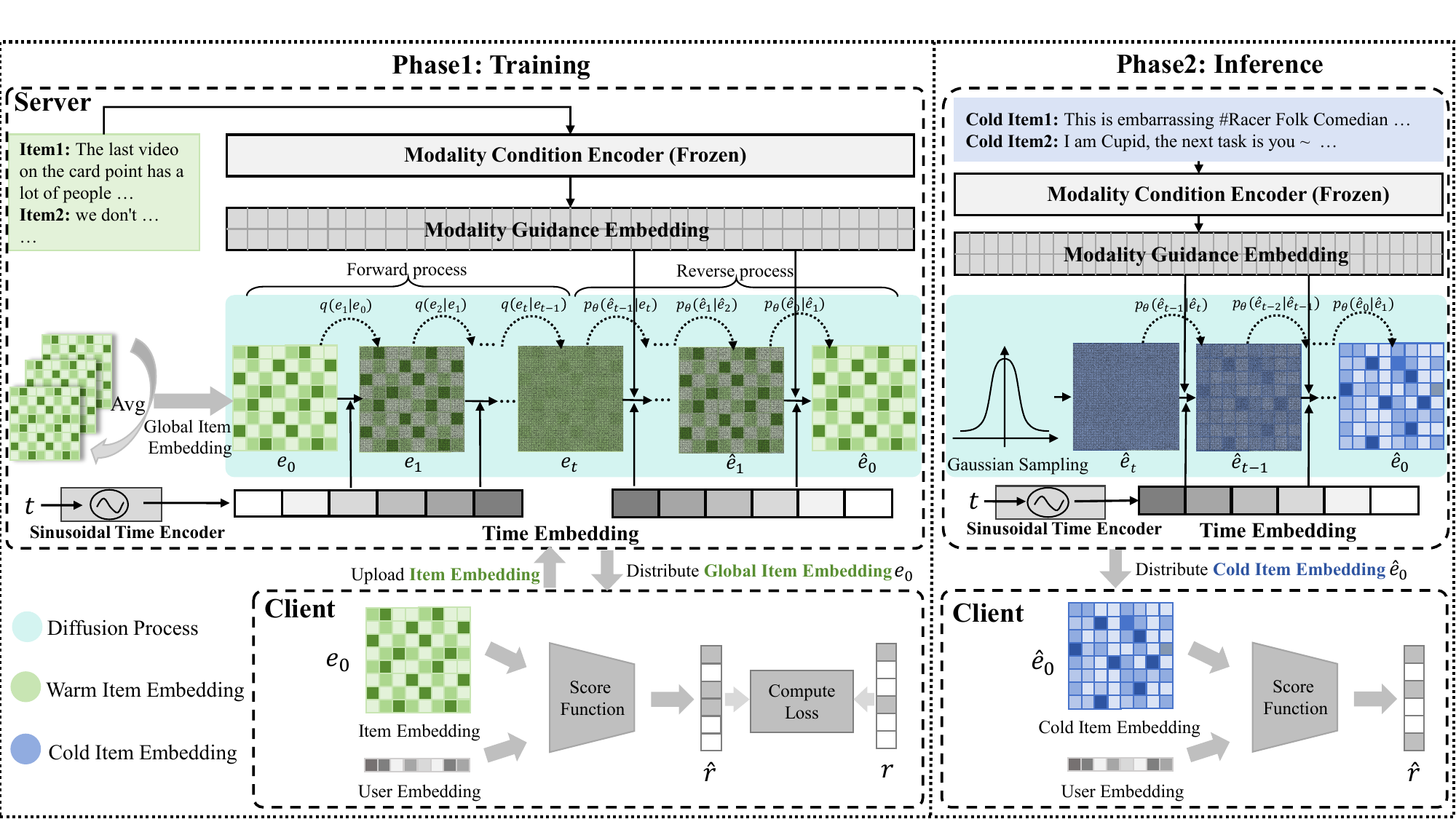} % Reduce the figure size so that it is slightly narrower than the column. Don't use precise values for figure width. This setup will avoid overfull boxes.
\caption{The overall framework of MDiffFR. We employ a modality encoder and a diffusion model on the server to extract modality features and learn the distribution of item embeddings, respectively. During the training phase, the aggregated global item embeddings are first distributed to each client for the next training round. Then, we add an iterative prior noise to the global item embeddings and encode modalities as conditions to guide the reverse denoising process of the diffusion model, ensuring alignment between the item embeddings and their semantics on the server. During the inference phase, we sample noise from a prior distribution and leverage encoded modality features to guide the reverse denoising process to generate semantically aligned item embeddings.}
\label{fig:framework}
\end{figure*}

\subsection{Modality-Guided Diffusion Processes}
During the training phase, the diffusion process on the server consists of two primary steps: the forward diffusion process injects noise into the embeddings, and the reverse denoising process leverages modality features as conditions to progressively reconstruct the original data.

\subsubsection{Forward Diffusion Process}
Given the global item embeddings \( \mathbf{E} \in \mathbb{R}^{m \times d} \), we set \( \bm{e}_0 = \mathbf{E} \) as the initial state, and iteratively add noise to \( \bm{e}_0 \) through the forward diffusion process. Under the Markov assumption, the state $\bm{e}_t$ at time step $t$ satisfies:
\begin{equation}
q(\bm{e}_t| \bm{e}_{t-1})=\mathcal{N}(\bm{e}_t;\sqrt{1-\beta_t} \bm{e}_{t-1},\beta_t \bm{I}),
\label{eq1}
\end{equation}
here, $t \in \{1, 2, \dots, T\}$, and $\beta_t$ denotes the variance of the Gaussian noise added at the $t$-th diffusion step. Leveraging the reparameterization trick and the additivity property of independent Gaussian distributions, the state $\bm{e}_t$ at any time step $t$ can be directly obtained from $\bm{e}_0$ as:
\begin{equation}
q(\bm{e}_t| \bm{e}_0)=\mathcal{N}(\bm{e}_t;\sqrt{\bar{\alpha}_t} \bm{e}_0,(1-\bar{\alpha}_t) \bm{I}),
\label{eq2}
\end{equation}
where $\alpha_t = 1 - \beta_t$, and $\bar{\alpha}_t = \prod_{s=1}^t \alpha_s$. Using the reparameterization trick \cite{DDPM}, we can express $\bm{e}_t$ as:
\begin{equation}
    \bm{e}_t = \sqrt{\bar{\alpha}_t} \, \bm{e}_0 + \sqrt{1 - \bar{\alpha}_t} \bm{\epsilon},
    \label{eq:compute_e_t}
\end{equation}
where $\bm{\epsilon} \sim \mathcal{N}(0, \bm{I})$. To control the amount of noise added at each step, we adopt the linear noise scheduler \cite{DiffRec} to regulate the noise scale: 
\begin{equation}
1-\bar{\alpha}_t=s\cdot\left[\alpha_{\min}+\frac{t-1}{T-1}(\alpha_{\max}-\alpha_{\min})\right],
\label{eq3}
\end{equation}
where $s$ denotes the noise scale factor, while $\alpha_{\max}$ and $\alpha_{\min}$ represent the maximum and minimum values of the noise scale, respectively.

\subsubsection{Reverse Denoising Process}
After obtaining $e_t$ through the forward diffusion process, the reverse denoising process starts from $e_t$ and iteratively reconstructs the original data by sampling the learned distributions. To enhance semantic alignment, MDiffFR leverages modality features to guide the reverse process, which is modeled as:
\begin{equation}
\small
p_\theta\left(\bm{e}_{t-1} | \bm{e}_t, \bm{m}\right)=\mathcal{N}\left(\bm{e}_{t-1};\mu_\theta\left(\bm{e}_t, t, \bm{m}\right),\Sigma_\theta\left(\bm{e}_t, t, \bm{m}\right)\right),
\label{eq4}
\end{equation}
where $\mu_\theta\left(\bm{e}_t, t, \bm{m}\right)$ and $\Sigma_\theta\left(\bm{e}_t, t, \bm{m}\right)$ denotes the mean the variance of the Gaussian distribution predicted by a neural network parameterized by $\theta$, which takes as input the noisy embedding $\bm{e}_t$, time step $t$, and modality condition $\bm{m}$. Accordingly, MDiffFR is trained by maximizing the marginal log-likelihood $\log p(\bm{e}_0 | \bm{m})$ of the target data, which is approximated by maximizing its evidence lower bound (ELBO):
\begin{equation}
\small
\begin{split}
    &\log p(\bm{e}_0 | \bm{m}) = \log \int p(\bm{e}_{0:T} | \bm{m}) \, d\bm{e}_{1:T} 
    \\&\geq \underbrace{\mathbb{E}_{q(\bm{e}_1 | \bm{e}_0)}\left[\log p_\theta(\bm{e}_0 | \bm{e}_1, \bm{m})\right]}_{\text{reconstruction term}}    \\
    & -\sum_{t=2}^T \underbrace{\mathbb{E}_{q(\bm{e}_t | \bm{e}_0)}\left[ D_{\text{KL}}\left( q(\bm{e}_{t-1} | \bm{e}_t, \bm{e}_0) \,\|\, p_\theta(\bm{e}_{t-1} | \bm{e}_t, \bm{m}) \right) \right]}_{\text{denoising matching term}}.
\end{split}
\label{eq:elbo}
\end{equation}
Here, the reconstruction term corresponds to the negative reconstruction error of recovering during the reverse denoising process. It measures how well the diffusion model can recover the embedding $\bm{e}_0$ from its noisy counterpart $\bm{e}_1$ under the conditional guidance of modality features $\bm{m}$. The denoising matching term reflects the iterative fitting of the true transition posterior $q(\bm{e}_{t-1}| \bm{e}_t,\bm{e}_0)$ at each step from $t=2$ to $T$ by the network parameterized by $\theta$ \cite{UnderstandDM}. It captures the gradual refinement of noisy embeddings across timesteps, leading to a smooth and stable generation process that produces semantically meaningful representations.

\subsection{Denoising Model Training}
The formulation of Eq. (\ref{eq:elbo}) shows that the training objective is primarily driven by two loss components: the denoising matching term $\mathcal{L}_{\text{dmt}}$ and the reconstruction term $\mathcal{L}_{\text{rt}}$.

\textbf{Estimation of denoising matching term.}
The form of the denoising matching term reveals that it encourages the distribution $ p_\theta(\bm{e}_{t-1} | \bm{e}_t,\bm{m})$ to approximate the true posterior $q(\bm{e}_{t-1} | \bm{e}_t, \bm{e}_0)$ via KL divergence. However, since directly modeling the true posterior $q(\bm{e}_{t-1} | \bm{e}_t, \bm{e}_0)$ is intractable, we reformulate the training objective using Bayes’ rule to obtain a tractable solution:
\begin{equation}
q\left(\bm{e}_{t-1}| \bm{e}_t, \bm{e}_0\right)\propto\mathcal{N}\left(\bm{e}_{t-1};\tilde{\mu}\left(\bm{e}_t, \bm{e}_0, t\right),\sigma^2(t) \bm{I}\right),
\label{eq6}
\end{equation}
here, $\tilde{\mu}\left(\bm{e}_t, \bm{e}_0, t\right)$ and $\sigma^2(t)\bm{I}$ represent the mean and standard deviation of the posterior distribution $q\left(\bm{e}_{t-1}| \bm{e}_t, \bm{e}_0\right)$, respectively, which can be further expressed as:
\begin{equation}
\begin{cases}
\begin{aligned}
&\tilde{\mu}(\bm{e}_t, \bm{e}_0, t)=\frac{\sqrt{\alpha_t}(1-\bar{\alpha}_{t-1})}{1-\bar{\alpha}_t} \bm{e}_t+\frac{\sqrt{\bar{\alpha}_{t-1}}(1-\alpha_t)}{1-\bar{\alpha}_t} \bm{e}_0, \\
&\sigma^2(t)=\frac{(1-\alpha_t)(1-\bar{\alpha}_{t-1})}{1-\bar{\alpha}_t}.
\end{aligned}
\end{cases}
\label{eq:miu_and_sigma}
\end{equation}
To further simplify the training process, we set $\Sigma_\theta(\bm{e}_t,t)=\sigma^2(t)\bm{I}$ directly \cite{DDPM}. Meanwhile, based on Eq. (\ref{eq:miu_and_sigma}), the posterior mean can be factorized as:
\begin{equation}
\small
\mu_\theta\left(\bm{e}_t, t, \bm{m}\right)=\frac{\sqrt{\alpha_t}\left(1-\bar{\alpha}_{t-1}\right)}{1-\bar{\alpha}_t} \bm{e}_t+\frac{\sqrt{\bar{\alpha}_{t-1}}\beta_t}{1-\bar{\alpha}_t}\hat{\bm{e}}_\theta\left(\bm{e}_t, t, \bm{m}\right).
\label{eq:miu_theta}
\end{equation}
By combining Eq. (\ref{eq:elbo}), Eq. (\ref{eq:miu_and_sigma}), and Eq. (\ref{eq:miu_theta}), the denoising matching term $\mathcal{L}_\text{dmt}$ is computed as:
\begin{equation}
\begin{aligned}
\mathcal{L}_{\text{dmt}} &= D_{\text{KL}}\left( q(\bm{e}_{t-1} | \bm{e}_t, \bm{e}_0) \,\|\, p_\theta(\bm{e}_{t-1} | \bm{e}_t, \bm{m}) \right)\\ 
&=\frac{1}{2}\left(\frac{\bar{\alpha}_{t-1}}{1-\bar{\alpha}_{t-1}}-\frac{\bar{\alpha}_t}{1-\bar{\alpha}_t}\right)\left\|\hat{\bm{e}}_\theta\left(\bm{e}_t, t,\bm{m}\right)-\bm{e}_0\right\|_2^2,
\label{eq:L_dmt}
\end{aligned}
\end{equation}
where $\bar{\alpha}_t$ denotes a predefined noise scheduling parameter, $\bm{e}_0$ is the target embedding, and $\hat{\bm{e}}_\theta(\bm{e}_t, t)$ represents the predicted embedding generated from the noisy input $\bm{e}_t$ at timestep $t$.

\textbf{Estimation of the reconstruction term.}
To simplify the optimization objective, we take the negative of the reconstruction term as follows $\mathcal{L}_{\text{rt}}$:
\begin{equation}
\begin{aligned}
    \mathcal{L}_{\text{rt}}&=-\mathbb{E}_{q(\bm{e}_1 | \bm{e}_0)}\left[\log p_\theta(\bm{e}_0 | \bm{e}_1,\bm{m})\right]\\
    &=E_{q\left(\bm{e}_1| \bm{e}_0\right)}\left[\left\|\hat{\bm{e}}_\theta\left(\bm{e}_1,1,\bm{m}\right)-\bm{e}_0\right\|_2^2\right].
\label{eq:L_rt}
\end{aligned}
\end{equation}
By combining Eq. (\ref{eq:elbo}), Eq. (\ref{eq:L_dmt}), and Eq. (\ref{eq:L_rt}), we obtain the final optimization objective, denoted as:
\begin{equation}
\begin{aligned}
\mathcal{L}_{\text{elbo}} &= - \mathcal{L_{\text{rt}}} - \Sigma_{t=2}^{T}{\mathcal{L}_{\text{dmt}}}  \\
&= \mathbb{E}_{t \sim \mathcal{U}(1, T)} \mathbb{E}_{q(\bm{e}_0)} \left[ \left\| \hat{\bm{e}}_{\theta}(\bm{e}_t, t,\bm{m}) - \bm{e}_0 \right\|_2^2 \right].
\label{eq:L_elbo}
\end{aligned}
\end{equation}

\subsection{Modality-Guided Strategy}
In the reverse denoising processes described above, we introduce a modality guidance condition $\bm{m}$ to ensure that the model not only captures the overall embedding distribution but also aligns the generated embeddings with the semantic information of items. Without such alignment, the generated embeddings fail to effectively represent cold-start items.

\textbf{Modality feature encoding.}
To enable modality features to serve as conditional guidance in the generative process, we deploy a modality condition encoder on the server to encode the modality information of items. In this work, we take textual modality as an example and adopt a pre-trained BERT model to extract modality features. Assuming that the textual modality of item $i$ consists of tokens $t_1^i,t_2^i,...,t_k^i$ , we obtain its modality feature representation as follows:
\begin{equation}
    {\bm{m}_i} = \mathcal{E}\left(\left\{ { \texttt{[CLS]};t_1^i,t_2^i,...,t_k^i} \right\} \right),
    \label{eq:extract_embedding}
\end{equation}
where $\mathcal{E}$ denotes the modality condition encoder implemented by a BERT model. The special token \texttt{[CLS]} is prepended to the token sequence, and its representation is used as the textual modality feature $\bm{m}_i$ of item $i$.

\textbf{Conditional dynamic fusion}
During the reverse denoising process, the importance of temporal and modality guidance conditions varies across different timesteps. In the early stage of reverse generation, the model needs to rapidly capture the overall embedding distribution from noise, making temporal information more crucial at this phase. As the generation progresses, the primary embedding distribution has already been captured, and the focus shifts toward achieving finer-grained semantic alignment between the embeddings and the corresponding items. Therefore, to dynamically capture the varying influence of guidance conditions during the reverse denoising process, we introduce a multi-head attention mechanism that adaptively fuses the guiding conditions across different timesteps. \cite{Attention}.
\begin{equation}
\mathbf{K} = \mathbf{V} = \text{Stack}\left[\mathcal{P}_t(\mathcal{T}_E(t)), \mathcal{P}_m(\bm{m})\right] 
\end{equation}
where the $\mathcal{T}_E(t)$ denotes a sinusoidal time encoding module; Its output is projected by a projection layer $\mathcal{P}_t$ to match the dimensionality of the guidance condition $\bm{m}$, which is processed by another projection layer $\mathcal{P}_m$. The projected representations are then stacked to construct the $\textbf{K}$ and $\textbf{V}$. Taking $\bm{e}_t$ as input to reconstruct $\bm{e}_{t-1}$, we have $head = \text{softmax}\left({\mathbf{Q}\mathbf{K}^\top}/{\sqrt{d/h}}\right),$
% \begin{equation}
%     \text{head} = \text{softmax}\left(\dfrac{\mathbf{Q}\mathbf{K}^\top}{\sqrt{D/h}}\right),
% \end{equation}
where $\mathbf{Q}=\bm{W}_q \cdot \bm{e}_t$ and $\bm{W}_q$ represent learnable parameters, $d$ is the dimensionality of the projected features, and $h$ represents the number of attention heads. The outputs from all attention heads are concatenated and projected to obtain the final fused representation, as shown in Eq. (\ref{eq:h_fused}).
\begin{equation}
    \bm{h}_\text{fused} = \text{Concat}(head_1,...,head_h)\cdot \bm{W}_o,
    \label{eq:h_fused}
\end{equation}
where $\bm{W}_o$ represents a learnable parameter matrix, and $\bm{h}_\text{fuse}$ denotes the fused representation, which serves as the final guidance signal for the reverse denoising process.

\subsection{Algorithm}
To better illustrate MDiffFR, we summarize the training and inference procedures procedures in Algorithm \ref{alg:algorithm_training} and \ref{alg:algorithm_inference}. 
 
As shown in Algorithm \ref{alg:algorithm_training}, we first utilize the modality encoder $\mathcal{E}$ to extract feature representations as conditional guidance during the training phase. Subsequently, the global item embeddings $\mathbf{E}$ are used as the initial state $\bm{e}_0$ to perform the forward and reverse diffusion processes, generating the predicted embedding $\hat{\bm{e}}_0$. The loss $\mathcal{L}_t$ is then computed and used to update the model parameters $\theta$.

\begin{algorithm}[h]
\centering
\caption{MDiffFR Training}
\label{alg:algorithm_training}
\begin{algorithmic}[1] %[1] enables line numbers
\Statex \textbf{Input}: participating clients $\mathcal{U}$; global rounds $R$; modality encoder $\mathcal{E}$; global parameters $\theta$; global item embedding $\mathbf{E}$;
\Statex \textbf{Server executes}:
\State Initial global embedding as $\mathbf{E}^1$;
\State Extract modality features $\mathbf{C}$ with Eq. (\ref{eq:extract_embedding}); 
\For{each round $r=1,2,...,R$}
    \If{$r > 1$}
        \State $\mathbf{E}^{r}  \leftarrow \text{Average}(\sum_{{u}=1}^n \mathbf{E}_{u}^{r-1})$;
    \EndIf
    \For{$i$ \textbf{in} server epoch}
        \State Initial $\bm{e}_0=\mathbf{E}^r$;
        \State Sample $t \sim \text{Uniform}(1,T)$;
        \State Compute $\bm{e}_t$ given $\bm{e}_0,t$ with Eq. (\ref{eq:compute_e_t});
        \State Compute $\mathcal{L}_{\text{elbo}}$ by Eq. (\ref{eq:L_elbo});
        \State Update $\theta$ based on $\nabla_\theta \mathcal{L}_t$;
    \EndFor
    \State Send $\mathbf{E}^r_{u}$ to each client ${u}$;
    \For{each client $u \in \mathcal{U}_p$ \textbf{in parallel}}
        \State $\mathbf{E}^{r+1}_{u} \leftarrow \mathrm{LocalTraining}   \left( \mathbf{E}_{u}^r, u \right)$ ;
    \EndFor
\EndFor
\end{algorithmic}
\end{algorithm}

During the inference phase, we first sample the data $\hat{\bm{e}}_t$ from a prior distribution, as illustrated in Algorithm \ref{alg:algorithm_inference}. Meanwhile, the cold-start item modality is encoded as a guidance condition for the generative process. MDiffFR then progressively reconstructs the data from noise through a denoising trajectory $\hat{\bm{e}}_t \rightarrow \hat{\bm{e}}_{t-1} \rightarrow \ldots \rightarrow \hat{\bm{e}}_0$. After obtaining the generated embedding $\hat{\bm{e}}_0$ for the new item, it is distributed to each client. Client $u$ then utilizes its local model to predict the rating score based on Eq. (\ref{eq:score_function}) and make recommendation decisions accordingly. The detailed procedure is provided in Algorithm \ref{alg:algorithm_inference}.

\begin{algorithm}[t]
\centering
\caption{MDiffFR Inference}
\label{alg:algorithm_inference}
\begin{algorithmic}[1]
\Statex \textbf{Input}: item modality.
\State Compute $\textbf{C}$ via Eq. (\ref{eq:extract_embedding})
\State Sample $\boldsymbol{\hat{\mathbf{\bm{e}}}_t} \sim \mathcal{N}(0, \mathbf{I})$
\For{$t = T, \dots, 1$}
    \State $\hat{\bm{e}}_{t-1} = \mu_\theta(\hat{\bm{e}}_t, t)$ calculated from $\hat{\bm{e}}_t$ via Eq. (\ref{eq:miu_theta})
\EndFor
\State Sends $\hat{\bm{e}}_0$ to each client $u$;
\end{algorithmic}
\end{algorithm}

\subsection{Light-MDiffFR}
To improve training efficiency, we reduce the number of training iterations by updating MDiffFR once every two rounds instead of training it in every round. This design provides a lightweight optimization that significantly improves the training efficiency of the model.

% \subsection{Algorithm}
% As illustrated in Algorithm 2, we also first encode the modality features of new items with Eq. (\ref{eq:extract_embedding}) during the inference phase. Then, we sample noise $\hat{\bm{e}}_t$ from a Gaussian prior distribution $\mathcal{N}$. Starting from the noisy input $\hat{\bm{e}}_t$, MDiffFR iteratively generates the target embedding $\hat{\bm{e}}_0$ through the reverse denoising process. Finally, we distribute the generated embeddings of cold-start items to each client for personalized recommendations.

\section{Privacy Guarantee Analysis}
In this section, we provide a theoretical comparison between mapping-based and generation-based models in terms of privacy leakage under adversarial inversion. 

In the presence of adversarial inversion attacks, we assume that an attacker can access a limited set of paired samples $\{(\bm{x}_i,\bm{y}_i)\}$, where $\bm{x}_i$ denotes the original modality feature and $\bm{y}_i$ is the corresponding item embedding.
Using these pairs, the attacker can train a shadow model to approximate the inverse mapping from embeddings $\bm{y}_i$ to modality features $\bm{x}_i$. Once the attacker obtains additional item embeddings from the system, the shadow model can be exploited to reconstruct the corresponding modality features, thereby recovering sensitive information from the embedding space. \cite{he2019model}.
For private modalities, such as expert-annotated commercial attributes, inversion attacks may lead to substantial commercial data leakage \cite{IFedRec}. For public modalities, the corresponding item embeddings often encode rich user–item interaction patterns across the entire system. As a result, unauthorized access to these embeddings can still reveal sensitive behavioral information, leading to significant privacy leakage \cite{FedMF,transFR}. Therefore, safeguarding the modality features is critical to preserving user privacy and preventing the risk of sensitive system-level information leakage.

% \subsection{Privacy Gap}
% To isolate the effect of stochasticity on privacy, we assume comparable marginal entropies between diffusion-based and deterministic embeddings, i.e., $H(\mathbf{E}{\mathrm{diff}}) \approx H(\mathbf{E}{\mathrm{map}})$.

\begin{lemma}[Inversion Risk Lower Bound via Mutual Information]
\label{lem:inversion_risk}
Let $\mathbf{X}$ be a discrete random variable taking $M$ possible values, and let $\mathbf{Y}$ denote its released embedding observed by an adversary. The adversary constructs an estimator $\hat{\mathbf{X}}(\mathbf{Y})$ to infer the original input $\mathbf{X}$, and the corresponding reconstruction error probability is defined as:
\begin{equation}
    P_e = \Pr(\hat{\mathbf{X}} \neq \mathbf{X}).
\end{equation}

Then, the reconstruction error $P_e$ is lower-bounded by the mutual information $I(\mathbf{X}; \mathbf{Y})$ as follows:
\begin{equation}
P_e \gtrapprox 1 - \frac{I(\mathbf{X}; \mathbf{Y}) + \log 2}{\log M}.
\label{eq:fano_approx}
\end{equation}
\end{lemma}

\begin{proof}
Let $\mathbf{X}$ be a discrete random variable supported on $M$ categories. Given the embedding $\mathbf{Y}$, the adversary attempts to reconstruct $\mathbf{X}$ through an estimator $\hat{\mathbf{X}}(\mathbf{Y})$. The probability of reconstruction error is:
\begin{equation}
    P_e = \Pr(\hat{\mathbf{X}} \neq \mathbf{X}).
\end{equation}

By Fano’s inequality, the conditional entropy of $\mathbf{X}$ given $\mathbf{Y}$ satisfies:
\begin{equation}
    H(\mathbf{X} \mid \mathbf{Y}) \leq H_b(P_e) + P_e \log(M - 1),
    \label{eq:fano_ineq}
\end{equation}
where $H_b(p)$ denotes the binary entropy function.  
From the definition of mutual information,
\begin{equation}
I(\mathbf{X}; \mathbf{Y}) = H(\mathbf{X}) - H(\mathbf{X} \mid \mathbf{Y}),
\end{equation}
which can be substituted into \eqref{eq:fano_ineq} yields:
\begin{equation}
I(\mathbf{X}; \mathbf{Y}) \geq H(\mathbf{X}) - H_b(P_e) - P_e \log(M - 1).
\label{eq:fano_MI}
\end{equation}

Rearranging terms gives an implicit lower bound on $P_e$:
\begin{equation}
    P_e \geq \frac{H(\mathbf{X})}{\log(M - 1)} - \frac{H_b(P_e) + I(\mathbf{X}; \mathbf{Y})}{\log(M - 1)}.
    \label{eq:implicit_bound}
\end{equation}

Since $H_b(P_e) \leq \log 2$, and assuming $\mathbf{X}$ follows a near-uniform distribution ($H(\mathbf{X}) \approx \log M$), we obtain the explicit form:
\begin{equation}
P_e \geq \frac{\log M}{\log(M - 1)} - \frac{I(\mathbf{X}; \mathbf{Y}) + \log 2}{\log(M - 1)}.
\end{equation}

For large $M$, we can approximate $\log(M-1) \approx \log M$, leading to:
\begin{equation}
P_e \gtrapprox 1 - \frac{I(\mathbf{X}; \mathbf{Y}) + \log 2}{\log M}.
\end{equation}

This result quantitatively connects the inversion risk with the mutual information between the original data $X$ and the corresponding output $Y$: lower $I(\mathbf{X}; \mathbf{Y})$ implies higher theoretical reconstruction error, thus indicating stronger privacy preservation.
\end{proof}

\begin{lemma}[Decomposition of Reverse Diffusion Process]
\label{lem:reverse_diffusion}
The reverse diffusion process can be approximated as the sum of a deterministic function of the modality input and an aggregated stochastic noise term, that is:
\begin{equation}
    \mathbf{E}_{\text{diff}} = s_T(\mathbf{x}_T, \mathbf{M}) + \mathbf{N}.
    \label{eq:diffusion_approx}
\end{equation}
where \( s_T(\mathbf{x}_T, \mathbf{M}) \) represents the deterministic generation trajectory conditioned on both modality \(\mathbf{M}\) and the initial input \(\mathbf{x}_T\), and \( \mathbf{N} \) denotes the accumulated Gaussian noise induced by the reverse diffusion process.

\end{lemma}

\begin{proof}
Consider the reverse diffusion process starting from an initial Gaussian noise sample \(\mathbf{x}_T \sim \mathcal{N}(0, \mathbf{I})\). At each step \(t = T, \dots, 1\), a denoising network \(\boldsymbol{\epsilon}_\theta(\mathbf{x}_t, t, \mathbf{M})\) predicts the noise component conditioned on the current state \(\mathbf{x}_t\) and the auxiliary input \(\mathbf{M}\). The transition from \(\mathbf{x}_t\) to \(\mathbf{x}_{t-1}\) is given by:
\begin{equation}
    \mathbf{x}_{t-1} = \frac{1}{\alpha_t}
    \left( 
        \mathbf{x}_t - \frac{1 - \alpha_t}{\sqrt{1 - \bar{\alpha}_t}} 
        \boldsymbol{\epsilon}_\theta(\mathbf{x}_t, t, \mathbf{M})
    \right) 
    + \sigma_t \mathbf{z}_t,
    \quad \mathbf{z}_t \sim \mathcal{N}(0, \mathbf{I}),
    \label{eq:reverse_diff}
\end{equation}
where \(\alpha_t\), \(\bar{\alpha}_t\), and \(\sigma_t\) are diffusion coefficients defined by the noise schedule.

Then, the reverse process can be decomposed into a deterministic component and a stochastic noise component:
\begin{equation}
    \mathbf{x}_{t-1} = s_t(\mathbf{x}_t, \mathbf{M}) + \sigma_t \mathbf{z}_t,
    \label{eq:reverse_step}
\end{equation}
where \(s_t(\mathbf{x}_t, \mathbf{M})\) denotes the deterministic mapping defined as:
\begin{equation}
    s_t(\mathbf{x}_t, \mathbf{M}) = 
    \frac{1}{\alpha_t}
    \left( 
        \mathbf{x}_t - \frac{1 - \alpha_t}{\sqrt{1 - \bar{\alpha}_t}} 
        \boldsymbol{\epsilon}_\theta(\mathbf{x}_t, t, \mathbf{M})
    \right).
\end{equation}

By recursively applying \eqref{eq:reverse_step} from \(t = T\) down to \(t = 1\), the final generated sample \(\mathbf{x}_0\) can be expressed as:
\begin{equation}
    \mathbf{x}_0 = s_T(\mathbf{x}_T, \mathbf{M}) + \sum_{t=1}^T \sigma_t \mathbf{z}_t.
    \label{eq:diffusion_final}
\end{equation}
Here, \(s_T(\mathbf{x}_T, \mathbf{M})\) represents the deterministic component of the initial input \(\mathbf{x}_T\) and the modality \(\mathbf{M}\), and \(\sum_{t=1}^T \sigma_t \mathbf{z}_t\) represents the accumulated stochastic noise.

Thus, the diffusion-generated embedding \(\mathbf{E}_{\text{diff}}\) can be approximated as:
\begin{equation}
    \mathbf{E}_{\text{diff}} = s_T(\mathbf{x}_T, \mathbf{M}) + \mathbf{N},
    \label{eq:diffusion_approx}
\end{equation}
where \(\mathbf{N} = \sum_{t=1}^T \sigma_t \mathbf{z}_t\) is the accumulated stochastic noise component.
\end{proof}

\begin{lemma}[Mutual Information of Mapping-based Models]
\label{lem:mapping}
Let $\mathcal{F}_\theta : \mathbf{X} \to \mathbf{E}$ be a mapping-based model and denote
$\mathbf{E}_{\mathrm{map}} := \mathcal{F}_\theta(\mathbf{X})$ on the training domain $\mathbf{X}_{\mathrm{train}}$.  
Then, for any $\mathbf{X} \in \mathbf{X}_{\mathrm{train}}$, the mutual information between the input and the corresponding embedding satisfies:
\begin{equation}
    I(\mathbf{X}; \mathbf{E}_{\mathrm{map}}) = H(\mathbf{E}_{\mathrm{map}}).
    \label{eq:mapping_mi}
\end{equation}
\end{lemma}

\begin{proof}
By assumption, the mapping $\mathcal{F}_\theta$ is approximately injective over the training domain, i.e., for any distinct inputs $\bm{x}_1, \bm{x}_2 \in \mathbf{X}_{\mathrm{train}}$,
\begin{equation}
    \| \mathcal{F}_\theta(\bm{x}_1) - \mathcal{F}_\theta(\bm{x}_2) \| \ge \delta, \quad \delta > 0.
\end{equation}

This implies that the mapping produces nearly unique embeddings for each training input. Therefore, there exists an approximate inverse $\mathcal{F}_\theta^{-1}$ such that:
\begin{equation}
    \mathcal{F}_\theta^{-1}(\mathcal{F}_\theta(\bm{x})) \approx \bm{x}, \quad \forall \bm{x} \in \mathbf{X}_{\mathrm{train}}.
\end{equation}

Consequently, the conditional entropy of the input given the embedding is:
\begin{equation}
    H(\mathcal{F}_\theta(\mathbf{X}) \mid \mathbf{X} ) = 0,
\end{equation}
because the injective nature of the mapping ensures that each embedding uniquely determines its corresponding input.
For simplicity of notation, we denote the model output $\mathcal{F}\theta(\mathbf{X})$ as $\mathbf{E}_{\mathrm{map}}$.
By the definition of mutual information, we have:
\begin{equation}
    I(\mathbf{X}; \mathcal{F}_\theta(\mathbf{X})) 
    = I(\mathbf{X}; \mathbf{E}_{\mathrm{map}}) 
    = H(\mathbf{E}_{\mathrm{map}}) - H(\mathbf{E}_{\mathrm{map}} \mid \mathbf{X}) 
    = H(\mathbf{E}_{\mathrm{map}}).
\end{equation}

Note that in this Lemma, the input $\mathbf{X}$ refers to the modality representation $\mathbf{M}$ of items. For consistency with the notations used in other Lemmas, we denote the model input as $\mathbf{X}$.
\end{proof}

\begin{lemma}[Mutual Information of Diffusion-based Models]
\label{lem:diffusion}
Let $\mathbf{M}$ denote the modality input, $p$ represents the dimension of item embeddings, and let $\mathbf{E}_{\mathrm{diff}}$ be the embedding generated by a diffusion model. Then, we have:
\begin{equation}
    I(\mathbf{M}; \mathbf{E}_{\mathrm{diff}})
    % = H(\mathbf{E}_{\mathrm{diff}}) - H(\mathbf{E}_{\mathrm{diff}} \mid \mathbf{M}) 
    \le H(\mathbf{E}_{\mathrm{diff}}) - h_{\min}, \quad \text{where } h_{\min} = \frac{p}{2} \log(2\pi e \, \sigma_{\min}^2) > 0.
\end{equation}
\end{lemma}

\begin{proof}
Due to the Gaussian noise injected at each reverse diffusion step (as shown in Lemma~\ref{lem:reverse_diffusion}), the final embedding can be written as:
\begin{equation}
    \mathbf{E}_{\mathrm{diff}} = s_T(x_T,\mathbf{M}) + \mathbf{N}, \quad \mathbf{N} \sim \mathcal{N}(0, \boldsymbol{\Sigma}), \quad \boldsymbol{\Sigma} \succeq \sigma_{\min}^2 \mathbf{I},
\end{equation}
where $s_T(x_T,\mathbf{M})$ is the deterministic component of the reverse process, and $\mathbf{N}$ represents the noise component. Therefore, the conditional entropy of the diffusion embedding given $\mathbf{M}$ is:
\begin{equation}
    H(\mathbf{E}_{\mathrm{diff}} \mid \mathbf{M}) = H(\mathbf{N}),
\end{equation}
and since $\mathbf{N}$ follows a Gaussian distribution with covariance $\boldsymbol{\Sigma} \succeq \sigma_{\min}^2 \mathbf{I}$, it is well-known from the maximum entropy property of Gaussian distributions that:
\begin{equation}
    H(\mathbf{N}) \ge h_{\min} := \frac{p}{2} \log(2\pi e \, \sigma_{\min}^2) > 0,
\end{equation}
where $h_{\min}$ represents the lower bound on the conditional entropy of $\mathbf{E}_{\mathrm{diff}}$ given $\mathbf{M}$, and $p$ is the embedding dimension. 

Substituting this lower bound into the definition of mutual information, we obtain:
\begin{equation}
    I(\mathbf{M}; \mathbf{E}_{\mathrm{diff}}) = H(\mathbf{E}_{\mathrm{diff}}) - H(\mathbf{E}_{\mathrm{diff}} \mid \mathbf{M}) \le H(\mathbf{E}_{\mathrm{diff}}) - h_{\min}.
\end{equation}

Thus, the mutual information between the input $\mathbf{M}$ and the diffusion-generated embedding $\mathbf{E}_{\mathrm{diff}}$ is upper-bounded by $H(\mathbf{E}_{\mathrm{diff}}) - h_{\min}$.
\end{proof}

\begin{remark}[Privacy Advantage of Diffusion-based Models over Deterministic Mapping]
\label{rem:privacy_advantage}
To isolate the impact of stochasticity, we assume comparable marginal entropies between the two embeddings, i.e., $H(\mathbf{E}_{\text{diff}}) \approx H(\mathbf{E}_{\text{map}})$. 
By combining Lemma~\ref{lem:mapping}–\ref{lem:diffusion}, we obtain the comparative privacy property between deterministic mapping models and diffusion-based generative models:
\begin{equation}
    I(\mathbf{M}; \mathbf{E}_{\text{diff}}) \le I(\mathbf{M}; \mathbf{E}_{\text{map}}) - h_{\min}.
\end{equation}

According to Lemma~\ref{lem:inversion_risk}, a smaller mutual information implies a lower inversion risk and thus stronger privacy preservation. 
Therefore, diffusion-based generative models inherently achieve better privacy protection than deterministic mapping models, as the stochastic diffusion process introduces irreducible uncertainty into the generated embeddings.
\end{remark}

\section{Experiments}
In this section, we conduct extensive experiments on four real-world datasets to answer the following research questions.
\begin{enumerate}
    \item Does MDiffFR achieve superior performance compared to state-of-the-art baselines?
    \item Can the generation-based diffusion method effectively alleviate the embedding misalignment between cold-start and warm items?
    \item Does modality guidance enhance the generative capability of the model?
    \item Does MDiffFR provide stronger privacy guarantees under practical attack scenarios?
    \item How does the key parameter affect the performance of the model?
\end{enumerate}

\subsection{Datasets}
In this paper, we conduct experiments on four real-world datasets collected from two online video platforms, KuaiShou and Bilibili, namely KU, Food, Dance, and Movie \cite{NineRec}. Table \ref{tab:data} shows the detailed statistics of these datasets. Here, Users correspond to different clients, Items represent different videos, and Interactions denote the number of interactions between users and items. Modality represents the number of modalities associated with the items. The datasets differ significantly in scale, with the number of users ranging from 2,034 to 16,525. In addition, to adapt to item cold-start scenarios and ensure a fair comparison, we partition these datasets into training, validation, and test sets using a 6:1:3 ratio, following the approaches outlined in \cite{Heater, IFedRec}. 

\begin{table}[t]
\label{tab:datasets}
\centering
\caption{The statistics of the used datasets.}
\begin{tabular}{lllll}
\toprule
Dataset & Users & Items & Interactions & Modality \\
\midrule
KU & 2,034 & 5,370 & 18,519 & 5,370 \\
Food & 6,549 & 1,579 & 39,740 & 1,579 \\
Dance & 10,715 & 2,307 & 83,392 & 2,307 \\
Movie & 16,525 & 3,509 & 115,576 & 3,509 \\
\bottomrule
\label{tab:data}
\end{tabular}
\end{table}

\subsection{Baselines}
For more fair comparison, we select four baselines in both federated and centralized recommendation settings, respectively. For federated recommendations (FRs), we choose pioneering works IFedRec addressing the item cold-start problem in FRs, including IPFedRec and IFedNCF. Additionally, we adapt two typical federated methods, FCF and FedNCF, to address the item cold-start problem, which are referred to as CS\_FCF and CS\_FedNCF. For centralized recommendations (CRs), we select three representative methods that address the cold-start problem: Heater, GAR, and ALDI. In addition, we include a diffusion-based model, DiffRec, which we adapt to the item cold-start scenario, referred to as CS\_DiffRec.

\begin{itemize}
    \item \textbf{CS\_FCF. \cite{FCF}} It is a first study in recommendation under federated settings, where the model is trained locally on clients and the parameters are uploaded to a central server for aggregation. In our experiment, we adapt it to address the item cold-start problem.

    \item \textbf{CS\_FedNCF. \cite{FedNCF}} It employs neural networks to capture the non-linear relationships between users and items on the client, effectively improving model performance in federated recommendation. In our experiments, we also adapt this approach to the item cold-start scenario.
    
    \item \textbf{IFedNCF. \cite{IFedRec}} It is the first solution to the item cold-start problem in federated recommendation, which integrates IFedRec into the FedNCF to address this challenge.
    
    \item \textbf{IPFedRec. \cite{IFedRec}} It is another variant of IFedRec designed to address the item cold-start problem in federated recommendation.
    
    \item \textbf{Heater. \cite{Heater}} It leverages a mixture-of-experts transformation mechanism and a randomized training strategy to enhance the effectiveness and accuracy of the model for cold-start recommendation.

    \item \textbf{GAR. \cite{GAR}} It is a generative model based on a generative adversarial network, which solves the cold-start problem by training a generator in an adversarial manner.

    \item \textbf{ALDI. \cite{ALDI}} It treats warm items as “teachers” and transfers their behavioral information to cold-start items, acting as “students”.

    \item \textbf{CS\_DiffRec. \cite{DiffRec}} It is one of the pioneering studies that employs a diffusion model to generate data for recommendations. We adapt it to address the item cold-start problem.
    
\end{itemize}

\begin{table*}[t]
\centering
\small
\setlength{\tabcolsep}{2pt} % 压缩列间距（默认6pt）
\renewcommand{\arraystretch}{1.1}
\caption{Performance comparison on four datasets with baselines in federated recommendation. The \textbf{best} performance is highlighted in bold, and the \underline{second} is underlined.}
\begin{tabular}{c|c|ccc|ccc|ccc|ccc}
\hline
\multirow{2}{*}{\textbf{Model}} &  \multirow{2}{*}{\textbf{Metric}} & \multicolumn{3}{c|}{\textbf{KU}} & \multicolumn{3}{c|}{\textbf{Food}} & \multicolumn{3}{c|}{\textbf{Dance}} & \multicolumn{3}{c}{\textbf{Movie}} \\
\cline{3-14}
& & @20 & @50 & @100 & @20 & @50 & @100 & @20 & @50 & @100 & @20 & @50 & @100 \\

\hline
\multirow{3}{*}{CS\_FedNCF} 
& Recall    & 1.12 & 3.92 & 8.12 & 4.96 & 16.08 & 29.61 & 4.26 & 12.63 & 24.75 & 2.53 & 8.20 & 15.00 \\
& Precision & 0.18 & 0.22 & 0.21 & 0.42 & 0.55 & 0.52 & 0.45 & 0.53 & 0.50 & 0.25 & 0.30 & 0.28 \\
& NDCG      & 0.88 & 1.82 & 2.97 & 2.00 & 5.67 & 8.29 & 2.42 & 5.23 & 8.32 & 1.41 & 3.19 & 4.82 \\
\cline{1-14}

 \multirow{3}{*}{CS\_FCF} 
& Recall    & 1.32 & 3.88 & 8.41 & 6.77 & 16.70 & 32.51 & 4.66 & 11.92 & 23.24 & 3.11 & 8.29 & 15.99 \\
& Precision & 0.19 & 0.19 & 0.21 & 0.59 & 0.58 & \textbf{0.57} & 0.47 & 0.48 & 0.47 & 0.29 & 0.30 & 0.29 \\
& NDCG      & 1.04 & 1.89 & 3.18 & 3.49 & 6.20 & 9.59 & 2.69 & 5.03 & 8.11 & 1.75 & 3.33 & 5.20 \\
\cline{1-14}

 \multirow{3}{*}{IPFedRec}
& Recall    & 7.05 & 14.51 & \underline{22.84} & 5.76 & 14.28 & 30.29 & \underline{8.57} & 17.69 & 31.73 & 4.28 & 10.31 & \underline{19.81} \\
& Precision & 0.72 & 0.58 & 0.48 & 0.45 & 0.45 & 0.47 & \underline{0.83} & 0.70 & \underline{0.63} & 0.39 & \underline{0.38} & 0.37 \\
& NDCG      & \underline{3.48} & 5.38 & \underline{7.49} & 3.15 & 5.62 & 9.48 & \underline{4.89} & 7.36 & 10.56 & 2.27 & 4.15 & 6.31 \\
\cline{1-14}

 \multirow{3}{*}{IFedNCF}
& Recall    & 3.77 & 8.39 & 14.78 & 6.47 & 16.18 & 32.25 & 5.95 & 15.18 & 28.50 & 3.54 & 8.45 & 16.73 \\
& Precision & 0.44 & 0.42 & 0.38 & 0.50 & 0.51 & 0.51 & 0.61 & 0.61 & 0.57 & 0.33 & 0.31 & 0.31 \\
& NDCG      & 2.29 & 4.08 & 6.01 & 3.55 & 6.59 & \underline{10.48} & 3.71 & 6.51 & 9.84 & 1.89 & 3.34 & 5.32 \\
\midrule \midrule

\multirow{3}{*}{\textbf{MDiffFR}}
& Recall    & \textbf{10.04} & \textbf{19.46} & \textbf{30.18} & \underline{7.50} & \textbf{19.23} & \textbf{35.92} & \textbf{10.29} & \textbf{21.37} & \textbf{35.19} & \underline{4.67} & \textbf{11.05} & \textbf{20.87} \\
& Precision & \textbf{1.08}  & \textbf{0.85}  & \textbf{0.67}  & \underline{0.60} & \textbf{0.62}  & \textbf{0.57}  & \textbf{1.05}  & \textbf{0.87}  & \textbf{0.71}  & \textbf{0.44} & \textbf{0.42} & \underline{0.39} \\
& NDCG      & \textbf{4.37}  & \textbf{7.07}  & \textbf{9.24}  & \underline{4.00} & \textbf{8.08}  & \textbf{11.76} & \textbf{5.94}  & \textbf{8.74}  & \textbf{11.76} & \textbf{2.60} & \textbf{4.49} & \textbf{6.82} \\
\cline{1-14}

 \multirow{3}{*}{\textbf{Light-MDiffFR}}
& Recall    & \underline{8.24} & \underline{18.24} & 22.21 & \textbf{8.78} & \underline{19.11} & \underline{34.16} & 6.30 & \underline{18.98} & \underline{34.98} & \textbf{4.80} & \underline{10.59} & \underline{19.15} \\
& Precision & \underline{0.91} & \underline{0.83} & \underline{0.51} & \textbf{0.69} & \underline{0.60} & \underline{0.54} & 0.67 & \underline{0.78} & \textbf{0.71} & \underline{0.41} & \underline{0.38} & \textbf{0.35} \\
& NDCG      & 3.35 & \underline{6.94} & 6.97 & \textbf{4.27} & \underline{7.26} & 10.27 & 4.17 & \underline{8.38} & \underline{11.85} & \underline{2.34} & \underline{4.29} & 6.11 \\
\bottomrule
\end{tabular}
\label{tab:FR_results}
\end{table*}

% \subsection{Experiment Settings}
% We conduct experiments on four real-world datasets \cite{NineRec}. We select several federated cold-start algorithms (IPFedRec, IFedNCF) \cite{IFedRec}, and additionally adapt two classical methods to the cold-start scenario in FR \cite{FCF, FedNCF}. To enable a more comprehensive comparison, we also select three advanced cold-start algorithms from Centralized Recommendation (\textbf{CR}) \cite{ALDI, GAR, Heater} and a diffusion-based model adapted to the cold-start task \cite{DiffRec}. In addition, we evaluate model performance using three metrics commonly used in FR: Recall, Precision, and NDCG. In all tables and figures, we report values using 1e-2 as the unit unless otherwise specified. 

\subsection{Implementation Details}
We conduct all experiments on four NVIDIA RTX A5000 GPUs using Python 3.8 and PyTorch v2.4.1 with CUDA 12.1. We set the maximum sequence length to 512 for BERT, truncating longer inputs and padding shorter ones with $\texttt{[PAD]}$ tokens. For both MDiffFR and all baseline models, we fix the embedding dimension to 64 and use a batch size of 256. During training, five negative items are sampled per interaction to evaluate model performance. In the federated setting, to ensure a fair comparison, we set the local learning rate to 0.1 for all models so that the global item embeddings remain consistent. For server learning, we use each method’s default value. Additionally, we configure the number of server epochs and the server batch size to 1 and 256, respectively. The client sampling ratio is set to 1.0, and the number of federated training rounds is fixed at 100. In addition, we evaluate model performance using three metrics commonly used in FRs: Recall, Precision, and NDCG. In all tables and figures, we report values using 1e-2 as the unit unless otherwise specified.

\subsection{Overall Performance (RQ1)}

\textbf{Compared to FRs.} Table \ref{tab:FR_results} presents a comparison of MDiffFR with baselines in FRs, and the results demonstrate that MDiffFR outperforms all baselines across all metrics. This validates the effectiveness of diffusion-based methods in capturing the underlying distribution of item embeddings, enabling the generation of higher-quality representations for cold-start items. Additionally, our method effectively mitigates the embedding misalignment issue, as detailed in the next section. Furthermore, Light-MDiffFR shows performance comparable to MDiffFR, outperforming most baselines and even surpassing MDiffFR in some cases, which further supports the potential applicability of our model in practical scenarios.

\begin{table*}[t]
\centering
\small
\setlength{\tabcolsep}{2pt} % 压缩列间距（默认6pt）
\renewcommand{\arraystretch}{1.1}
\caption{Performance comparison on four datasets with baselines in centralized recommendation. }
\begin{tabular}{c|c|ccc|ccc|ccc|ccc}
\hline
\multirow{2}{*}{\textbf{Model}} &  \multirow{2}{*}{\textbf{Metric}} & \multicolumn{3}{c|}{\textbf{KU}} & \multicolumn{3}{c|}{\textbf{Food}} & \multicolumn{3}{c|}{\textbf{Dance}} & \multicolumn{3}{c}{\textbf{Movie}} \\
\cline{3-14}
& & @20 & @50 & @100 & @20 & @50 & @100 & @20 & @50 & @100 & @20 & @50 & @100 \\

\hline

\multirow{3}{*}{Heater}
& Recall    & 2.22 & 4.53 & 8.98 & 6.89 & 12.08 & 32.72 & \underline{8.62} & 15.66 & 29.19 & 3.42 & 10.69 & \textbf{23.99} \\
& Precision & 0.28 & 0.23 & 0.23 & 0.56 & 0.48 & 0.52 & \underline{0.91} & 0.64 & \underline{0.60} & 0.34 & 0.41 & \textbf{0.45} \\
& NDCG      & 1.14 & 1.70 & 2.59 & 2.51 & 4.34 & 7.48 & 3.17 & 4.73 & 7.38 & 1.45 & 3.12 & 5.64 \\
\cline{1-14}

\multirow{3}{*}{GAR}
& Recall    & 0.30 & 1.48 & 2.82 & 2.68 & 9.27 & 17.77 & 3.42 & 6.11 & 12.59 & 1.42 & 3.36 & 8.74 \\
& Precision & 0.05 & 0.09 & 0.08 & 0.20 & 0.27 & 0.26 & 0.36 & 0.26 & 0.28 & 0.14 & 0.13 & 0.16 \\
& NDCG      & 0.14 & 0.45 & 0.71 & 0.92 & 2.32 & 3.88 & 1.14 & 1.78 & 3.06 & 0.42 & 0.86 & 1.86 \\
\cline{1-14}

\multirow{3}{*}{ALDI}
& Recall    & 0.74 & 8.64 & 10.25 & 2.21 & 6.32 & 11.75 & 3.43 & 6.02 & 11.29 & 1.21 & 5.93 & 11.97 \\
& Precision & 0.10 & 0.38 & 0.23 & 0.17 & 0.19 & 0.17 & 0.37 & 0.26 & 0.25 & 0.11 & 0.23 & 0.23 \\
& NDCG      & 0.54 & 2.69 & 3.00 & 0.78 & 1.66 & 2.60 & 1.10 & 1.69 & 2.72 & 0.51 & 1.57 & 2.75 \\
\cline{1-14}

\multirow{3}{*}{CS\_DiffRec}
& Recall    & 0.03 & 0.15 & 0.46 & 0.07 & 0.53 & 1.74 & 0.16 & 0.32 & 0.96 & 0.21 & 0.59 & 1.25 \\
& Precision & 0.01 & 0.02 & 0.02 & 0.01 & 0.02 & 0.03 & 0.02 & 0.02 & 0.02 & 0.02 & 0.02 & 0.02 \\
& NDCG      & 0.02 & 0.05 & 0.12 & 0.02 & 0.12 & 0.33 & 0.07 & 0.11 & 0.24 & 0.08 & 0.17 & 0.29 \\

\midrule \midrule

\multirow{3}{*}{\textbf{MDiffFR}}
& Recall    & \textbf{10.04} & \textbf{19.46} & \textbf{30.18} & \underline{7.50} & \textbf{19.23} & \textbf{35.92} & \textbf{10.29} & \textbf{21.37} & \textbf{35.19} & \underline{4.67} & \textbf{11.05} & \underline{20.87} \\
& Precision & \textbf{1.08}  & \textbf{0.85}  & \textbf{0.67}  & \underline{0.60} & \textbf{0.62}  & \textbf{0.57}  & \textbf{1.05}  & \textbf{0.87}  & \textbf{0.71}  & \textbf{0.44} & \textbf{0.42} & \underline{0.39} \\
& NDCG      & \textbf{4.37}  & \textbf{7.07}  & \textbf{9.24}  & \underline{4.00} & \textbf{8.08}  & \textbf{11.76} & \textbf{5.94}  & \textbf{8.74}  & \underline{11.76} & \textbf{2.60} & \textbf{4.49} & \textbf{6.82} \\
\cline{1-14}

\multirow{3}{*}{\textbf{Light-MDiffFR}}
& Recall    & \underline{8.24} & \underline{18.24} & 22.21 & \textbf{8.78} & \underline{19.11} & \underline{34.16} & 6.30 & \underline{18.98} & \underline{34.98} & \textbf{4.80} & \underline{10.59} & 19.15 \\
& Precision & \underline{0.91} & \underline{0.83} & \underline{0.51} & \textbf{0.69} & \underline{0.60} & \underline{0.54} & 0.67 & \underline{0.78} & \textbf{0.71} & \underline{0.41} & \underline{0.38} & 0.35 \\
& NDCG      & \underline{3.35} & \underline{6.94} & \underline{6.97} & \textbf{4.27} & \underline{7.26} & \underline{10.27} & \underline{4.17} & \underline{8.38} & \textbf{11.85} & \underline{2.34} & \underline{4.29} & \underline{6.11} \\
\bottomrule
\end{tabular}
\label{tab:CR_results}

\end{table*}
\textbf{Compared to CRs.} As shown in Table \ref{tab:CR_results}, our method outperforms the baselines in CRs for most cases, while ensuring privacy guarantees through federated settings. The only exception occurs in Recall@100 and Precision@100 of the Movie dataset, where the performance of MDiffFR is marginally lower than the Heater. We attribute this to Heater's centralized design that leverages user personal information as auxiliary data for cold-start recommendations. However, such personal information constitutes protected privacy under federated learning frameworks, making it inherently inaccessible in our method. Notably, CS\_DiffRec, which directly generates cold-start item interactions, exhibits the worst performance. This demonstrates that such direct interaction generation is infeasible.

\begin{figure*}[t]
    \centering
    % 第一行
    \begin{subfigure}{0.19\textwidth}
        \includegraphics[width=\linewidth]{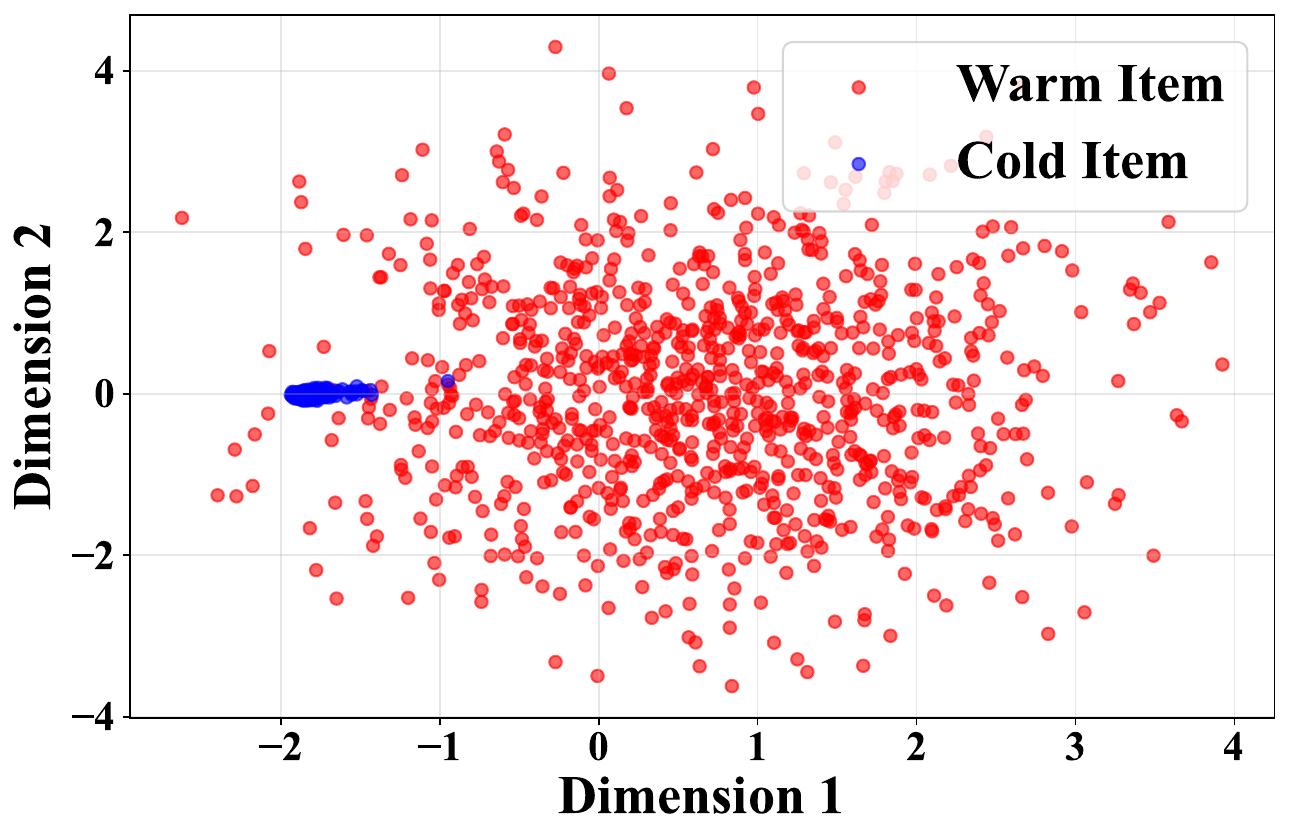}
        \caption{Round 0}
    \end{subfigure}
    \begin{subfigure}{0.19\textwidth}
        \includegraphics[width=\linewidth]{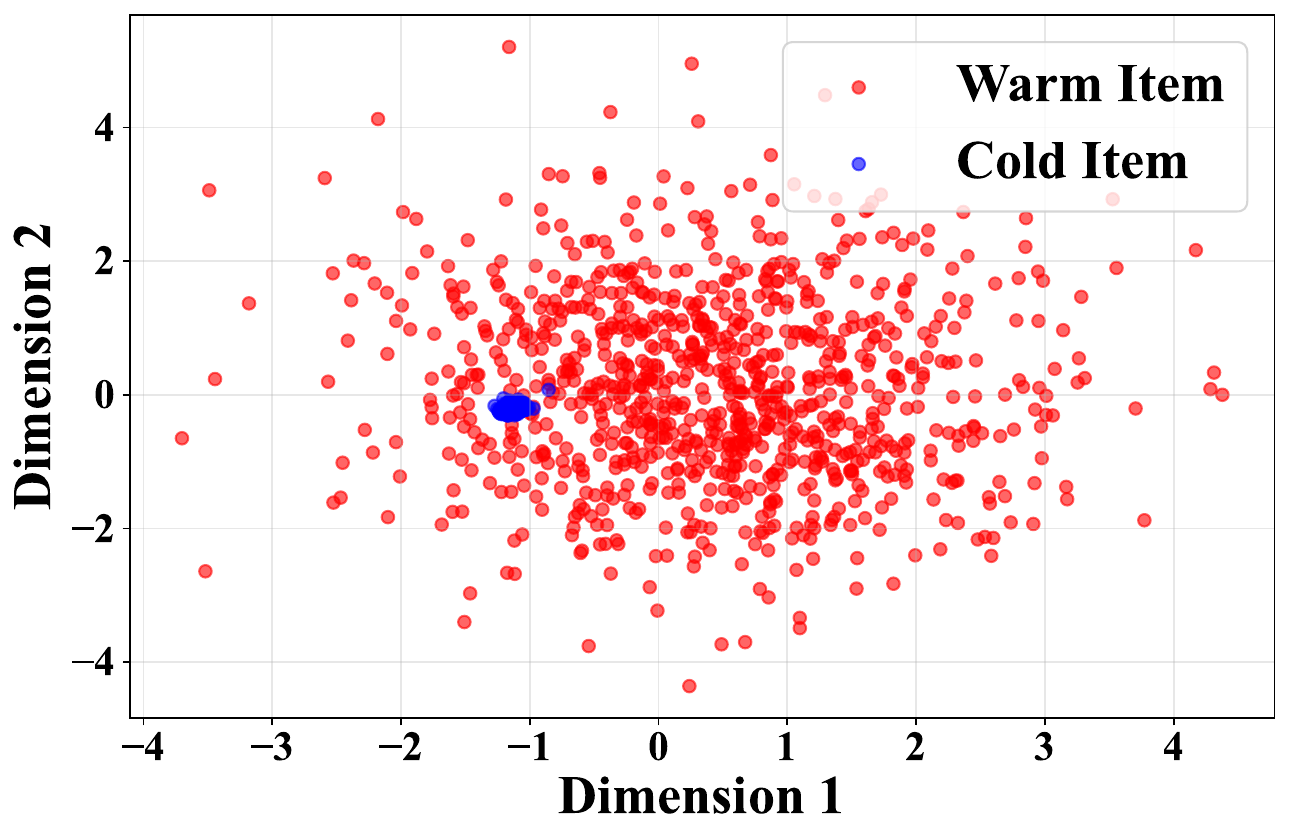}
        \caption{Round 10}
    \end{subfigure}
    \begin{subfigure}{0.19\textwidth}
        \includegraphics[width=\linewidth]{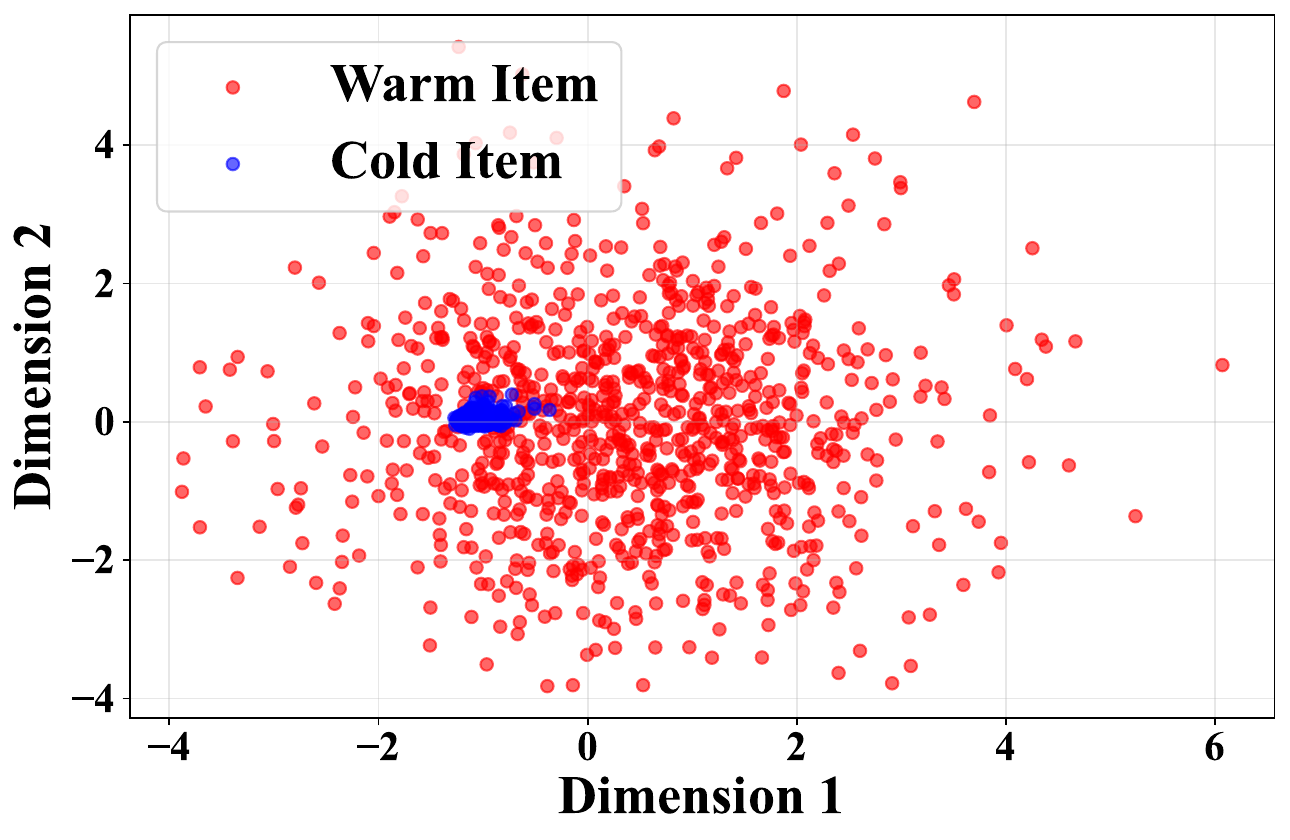}
        \caption{Round 20}
    \end{subfigure}
    \begin{subfigure}{0.19\textwidth}
        \includegraphics[width=\linewidth]{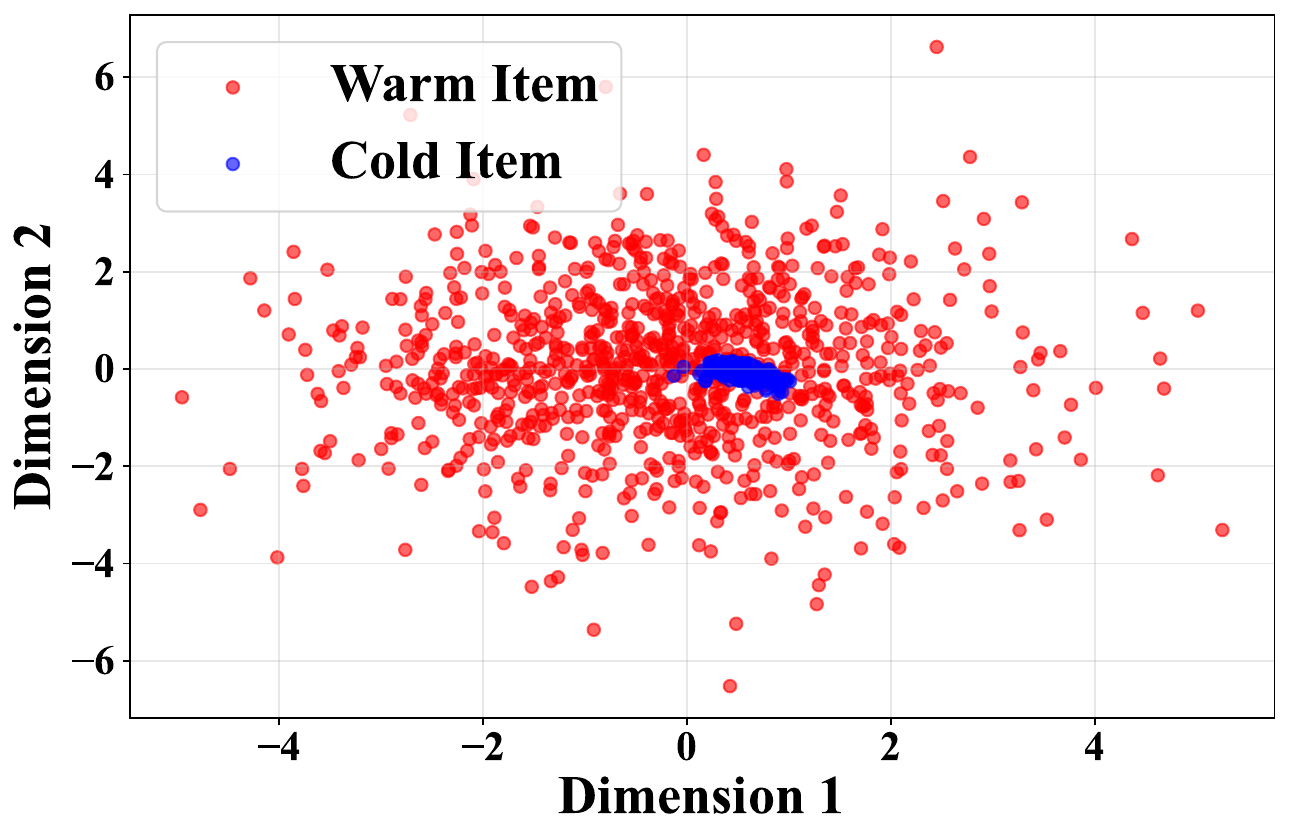}
        \caption{Round 30}
    \end{subfigure}
    \begin{subfigure}{0.19\textwidth}
        \includegraphics[width=\linewidth]{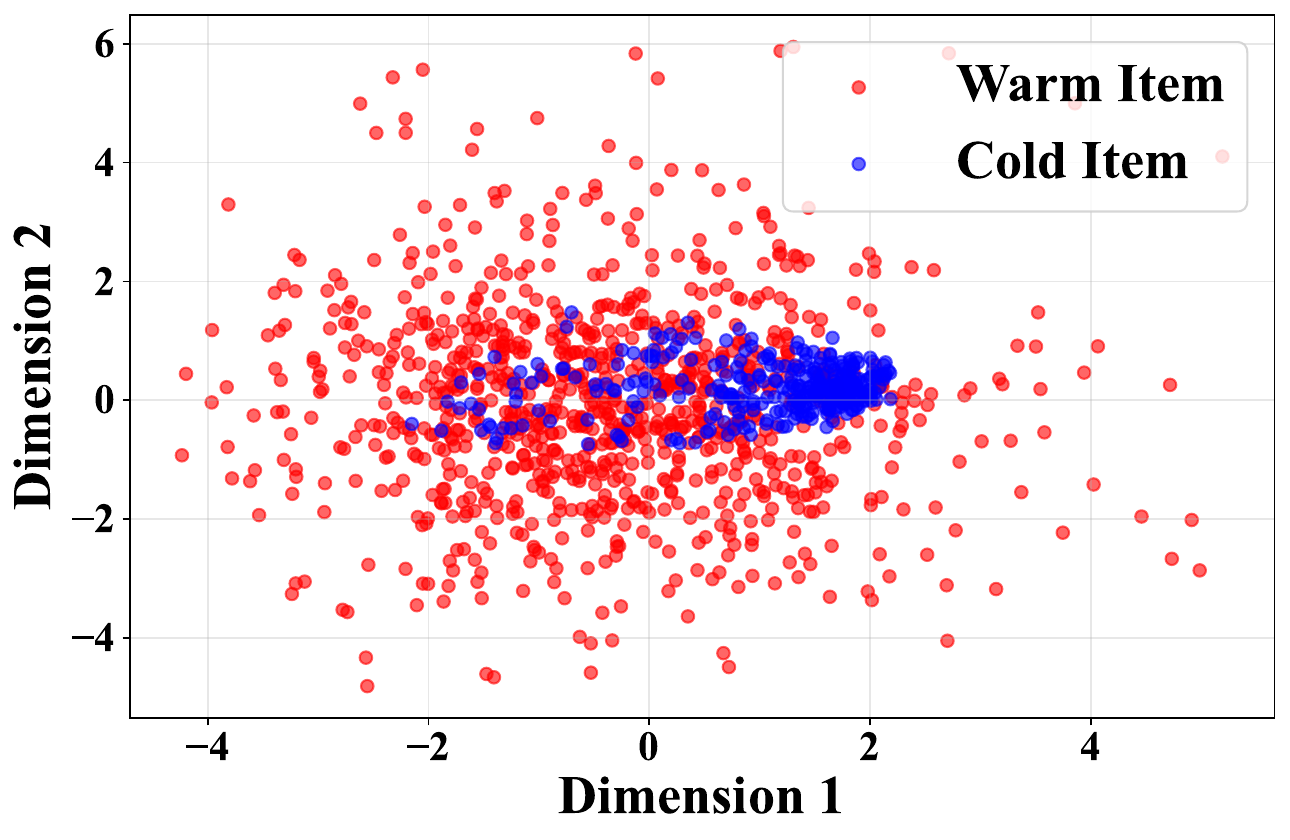}
        \caption{Round 40}
    \end{subfigure}

    % 第二行
    \begin{subfigure}{0.19\textwidth}
        \includegraphics[width=\linewidth]{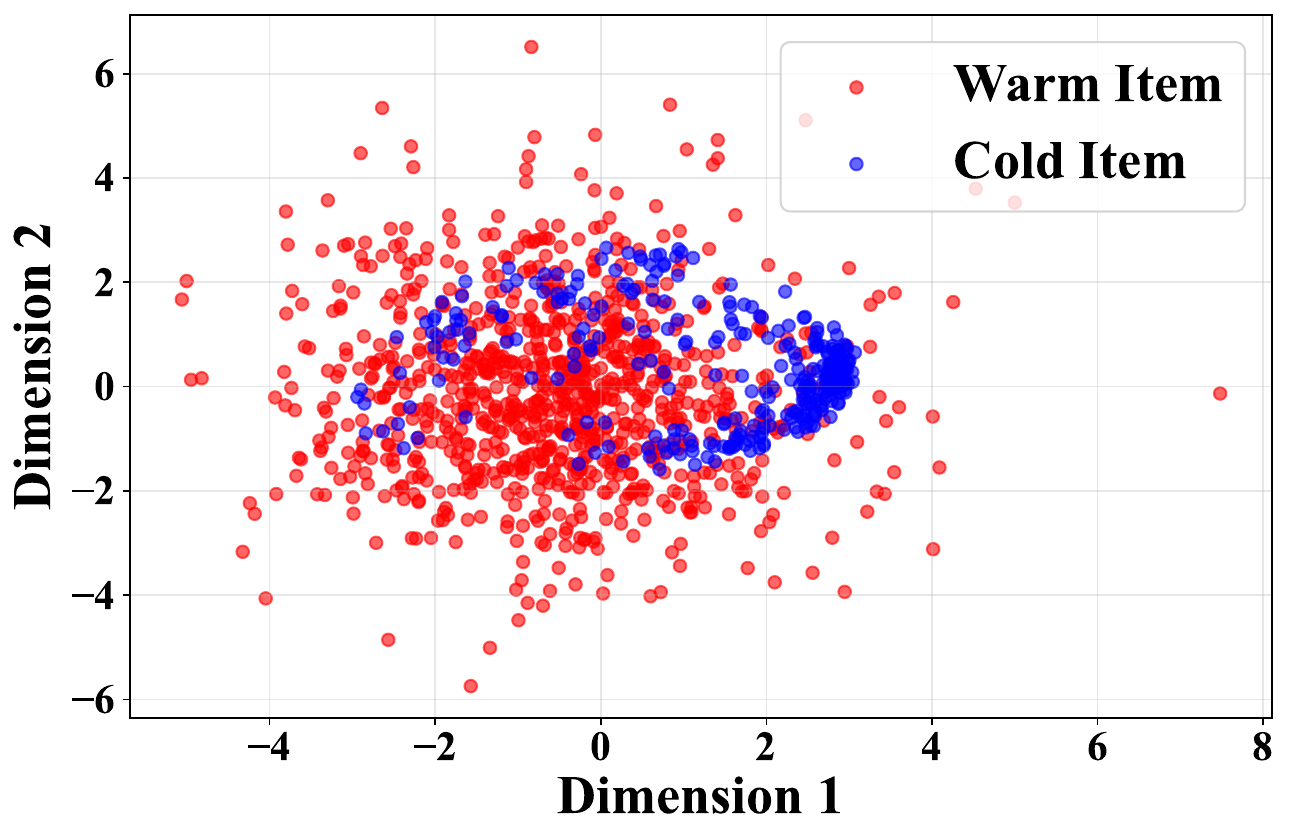}
        \caption{Round 50}
    \end{subfigure}
    \begin{subfigure}{0.19\textwidth}
        \includegraphics[width=\linewidth]{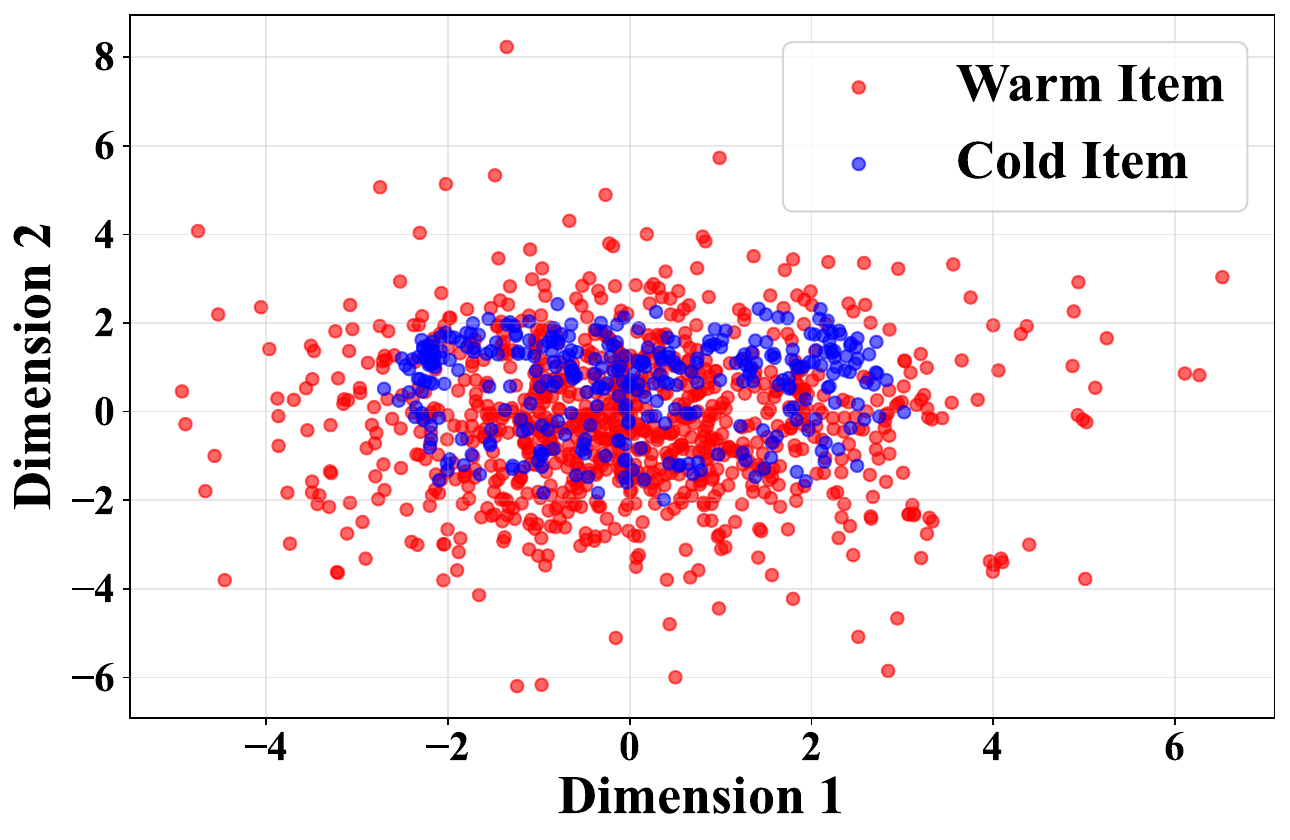}
        \caption{Round 60}
    \end{subfigure}
    \begin{subfigure}{0.19\textwidth}
        \includegraphics[width=\linewidth]{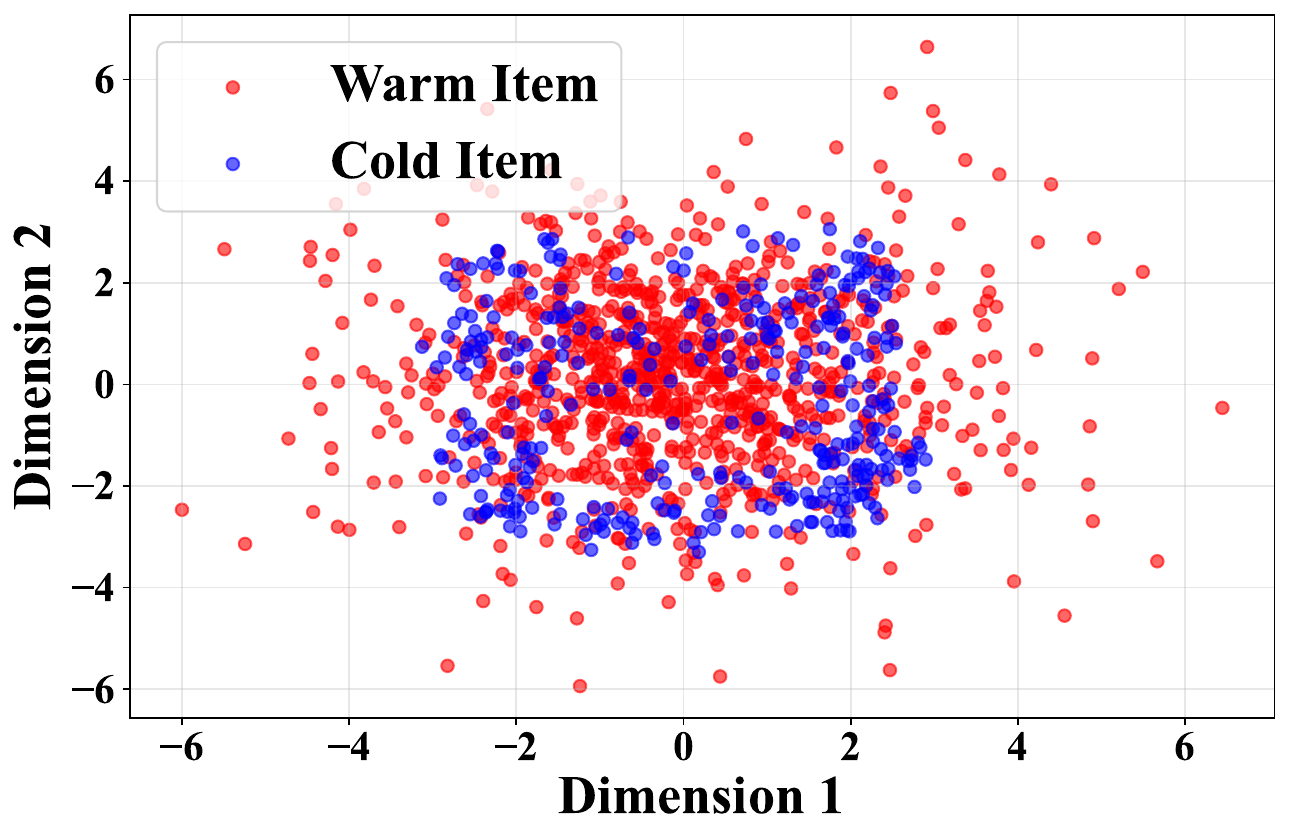}
        \caption{Round 70}
    \end{subfigure}
    \begin{subfigure}{0.19\textwidth}
        \includegraphics[width=\linewidth]{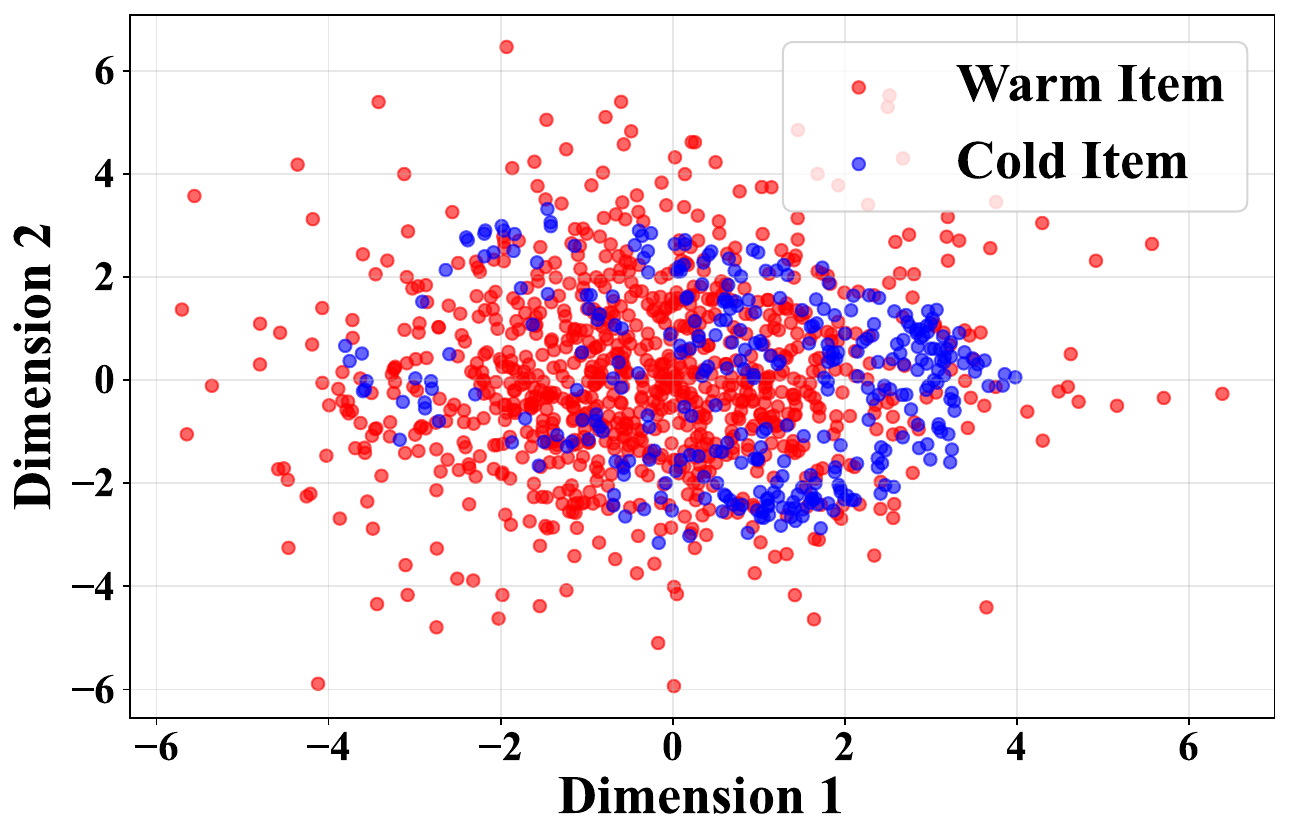}
        \caption{Round 80}
    \end{subfigure}
    \begin{subfigure}{0.19\textwidth}
        \includegraphics[width=\linewidth]{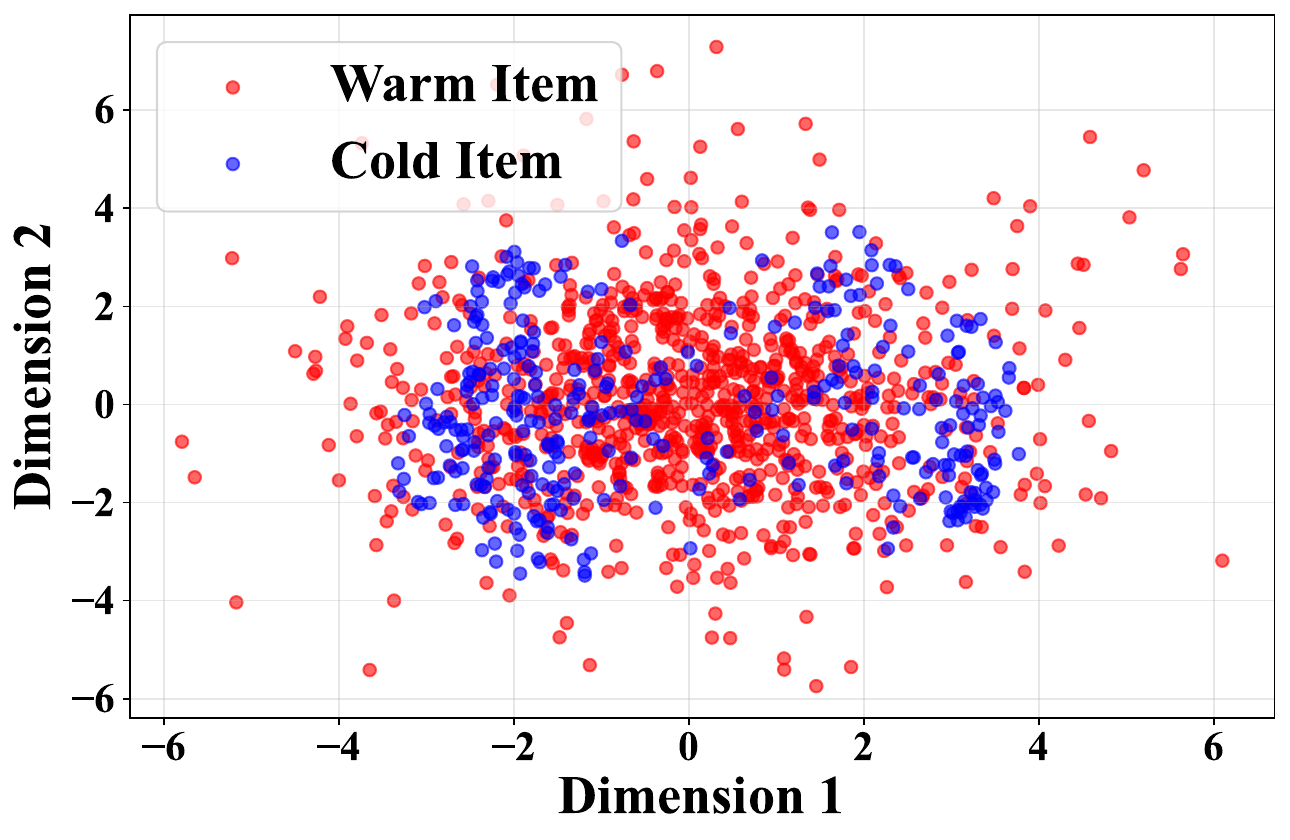}
        \caption{Round 90}
    \end{subfigure}

    \caption{The evolution of cold-start item embeddings during training in the Food dataset. Early in training, these embeddings deviate from the warm item distribution, but progressively shift toward and eventually align with the overall warm item embedding distribution.}
    \label{fig:embedding_distribution}
\end{figure*}
\subsection{Embedding Distribution (RQ2)} 
As illustrated in Fig. \ref{fig:embedding_distribution}, we visualize the distribution of embeddings for warm and cold-start items across training rounds. Initially, the embeddings generated by MDiffFR for cold-start items deviate significantly from the distribution center of warm items. As training progresses, these embeddings gradually approach the warm item distribution. By round 30, the generated embeddings are centered around the true distribution, and from round 40 onward, the cold-start embeddings increasingly align with the overall distribution of warm items. By round 60, the generated cold-start embeddings effectively capture the underlying distribution represented by the warm item embeddings, thereby alleviating the embedding misalignment issue observed in mapping-based methods.

\begin{table*}[t]
\centering
\small
\setlength{\tabcolsep}{2pt}
\renewcommand{\arraystretch}{1.1}
\caption{Ablation study evaluating the effectiveness of modality guidance in MDiffFR. \textbf{w/o Guidance} and \textbf{w/o Attention} denote removing the guidance condition and attention-based dynamic fusion mechanism, respectively, while \textbf{w/ Zero} and \textbf{w/ Rand} indicate substituting the guidance condition with zero and random vectors, respectively.}
\begin{tabular}{l|c|ccc|ccc|ccc|ccc}
\hline
\multirow{2}{*}{\textbf{Model}} & \multirow{2}{*}{\textbf{Metric}} & \multicolumn{3}{c|}{\textbf{KU}} & \multicolumn{3}{c|}{\textbf{Food}} & \multicolumn{3}{c|}{\textbf{Dance}} & \multicolumn{3}{c}{\textbf{Movie}} \\
\cline{3-14}
& & @20 & @50 & @100 & @20 & @50 & @100 & @20 & @50 & @100 & @20 & @50 & @100 \\
\hline

\multirow{3}{*}{w/o Guidance}
& Recall    & 2.13 & 4.78 & 8.54  & 5.26 & 14.55 & 30.75 & 4.90 & 11.25 & 28.57 & 3.82 & 9.88 & 19.38 \\
& Precision & 0.26 & 0.24 & 0.22  & 0.41 & 0.46  & 0.49  & 0.50 & 0.48  & 0.59  & 0.38 & 0.38 l& 0.37 \\
& NDCG      & 1.42 & 2.30 & 3.31  & 2.89 & 5.68  & 9.50  & 2.59 & 4.85  & 9.79  & 2.20 & 4.06 & 6.37 \\
\cline{1-14}

\multirow{3}{*}{w/o Attention}
& Recall    & 1.38 & 10.09 & 13.20 & 6.57 & 16.35 & 33.56 & 8.25 & 15.41 & 28.37 & 2.94 & 8.12 & 17.22 \\
& Precision & 0.20 &  0.49 &  0.34 & 0.53 &  0.52 &  0.53 & 0.80 &  0.61 &  0.58 & 0.28 & 0.31 &  0.32 \\
& NDCG      & 0.88 &  3.68 &  4.37 & 4.42 &  6.90 & 10.89 & 4.48 &  6.02 &  9.15 & 1.74 & 3.31 &  5.60 \\
\cline{1-14}

\multirow{3}{*}{w/ Zero}
& Recall    & 8.44 & 10.71 & 16.81 & 7.08 & 17.24 & 34.27 & 6.22 & 14.46 & 27.20 & 4.48 & 10.02 & 18.61 \\
& Precision & 0.94 & 0.50  & 0.42  & 0.58 & 0.56  & 0.55  & 0.65 & 0.61  & 0.56  & 0.42 & 0.37 & 0.35 \\
& NDCG      & 4.53 & 4.51  & 6.33  & 3.61 & 6.83  & 10.61 & 3.44 & 6.09  & 9.03  & 2.42 & 4.08 & 6.01 \\
\cline{1-14}

\multirow{3}{*}{w/ Rand}
& Recall    & 1.54 & 12.30 & 29.95 & 6.81 & 17.49 & 34.69 & 8.12 & 19.15 & 33.31 & 3.55 & 9.47 & 18.80 \\
& Precision & 0.24 & 0.58  & 0.67  & 0.55 & 0.56  & 0.56  & 0.80 & 0.76  & 0.67  & 0.34 & 0.36 & 0.35 \\
& NDCG      & 1.08 & 4.43  & 8.95  & 3.45 & 6.81  & 11.11 & 4.79 & 8.00  & 11.16 & 1.70 & 3.47 & 5.74 \\
\cline{1-14}

\multirow{3}{*}{MDiffFR}
& Recall    & \textbf{10.04} & \textbf{19.46} & \textbf{30.18} & \textbf{7.50} & \textbf{19.23} & \textbf{35.92} & \textbf{10.29} & \textbf{21.37} & \textbf{35.19} & \textbf{4.67} & \textbf{11.05} & \textbf{20.87} \\
& Precision & \textbf{1.08}  & \textbf{0.85}  & \textbf{0.67}  & \textbf{0.60} & \textbf{0.62}  & \textbf{0.57}  & \textbf{1.05}  & \textbf{0.87}  & \textbf{0.71}  & \textbf{0.44} & \textbf{0.42} & \textbf{0.39} \\
& NDCG      & \textbf{4.37}  & \textbf{7.07}  & \textbf{9.24}  & \textbf{4.00} & \textbf{8.08}  & \textbf{11.76} & \textbf{5.94}  & \textbf{8.74}  & \textbf{11.76} & \textbf{2.60} & \textbf{4.49} & \textbf{6.82} \\

\bottomrule
\end{tabular}
\label{tab:ablation-results}

\end{table*}

\subsection{Ablation Study (RQ3)}

We conducted ablation experiments to explore the effectiveness of modality-guided learning. As shown in Table \ref{tab:ablation-results}, the absence of guidance conditions results in a substantial decline in performance (e.g., Recall@20 drops by 79\% on the KU dataset), highlighting the importance of semantic alignment in embedding generation. Additionally, removing the attention-based dynamic fusion mechanism also results in a significant performance degradation, which empirically establishes that the integration of adaptive modality guidance is critical for achieving semantically consistent embedding generation. Furthermore, using non-informative conditions as guidance (e.g., zero or random vectors) degrades model performance, demonstrating that the quality of modality guidance directly impacts the effectiveness of the diffusion-based embedding reconstruction.

\subsection{Privacy Analysis (RQ4)}
To evaluate the privacy guarantee of MDiffFR, we design an inversion attack experiment to reconstruct the original modality embeddings. Consistent with the privacy analysis setting, we assume that the attacker has access to partial original modality information as well as all distributed embeddings of cold-start items. In our experiment, we assume that 20\% of the original modality embeddings are leaked, and we construct a multilayer neural network as the attack function $f^{-1}_\theta(y)$, where $y=f_\theta(x)$ and $x$ denotes the embedding of the input modality. The attacker aims to reconstruct an approximation \(\hat{\bm{x}}\) of the original input \(\bm{x}\) by minimizing the reconstruction error between the estimated and true inputs. The corresponding optimization objective is formulated as:
\begin{equation}
    \arg\min_{\theta} \| f^{-1}_\theta(\bm{y}) - \bm{x} \|^2,
\end{equation}

We adopt Mean Squared Error (MSE), Mean Absolute Error (MAE), Cosine Similarity (Cosine), and Pearson Correlation Coefficient (Pearson) as evaluation metrics to assess the effectiveness of the privacy guarantee. Specifically, MSE and MAE measure the reconstruction error between the recovered embeddings and the original modality information, where higher values indicate stronger privacy protection. In contrast, Cosine and Pearson quantify the similarity between the reconstructed and original embeddings, where larger values reflect higher privacy leakage. As shown in Fig. \ref{fig:Radar_Chart}, MDiffFR consistently achieves nearly zero Pearson and Cosine similarities across all four datasets, suggesting that the vectors generated by inversion attacks are almost uncorrelated with the true modality features. Only the KU dataset shows a marginally positive Cosine similarity, but it remains considerably lower than that of IPFedRec. In terms of MSE and MAE, MDiffFR yields higher values than IPFedRec, indicating greater discrepancy between the reconstructed vectors and the original inputs. 
These results demonstrate that MDiffFR offers significantly stronger privacy protection and effectively mitigates information leakage under inversion attacks.

\begin{figure}[h]
    \centering
    \begin{minipage}{0.85\columnwidth} % 必须用 minipage 包住
    \centering
    \begin{subfigure}[t]{0.24\columnwidth}
        \includegraphics[width=\linewidth]{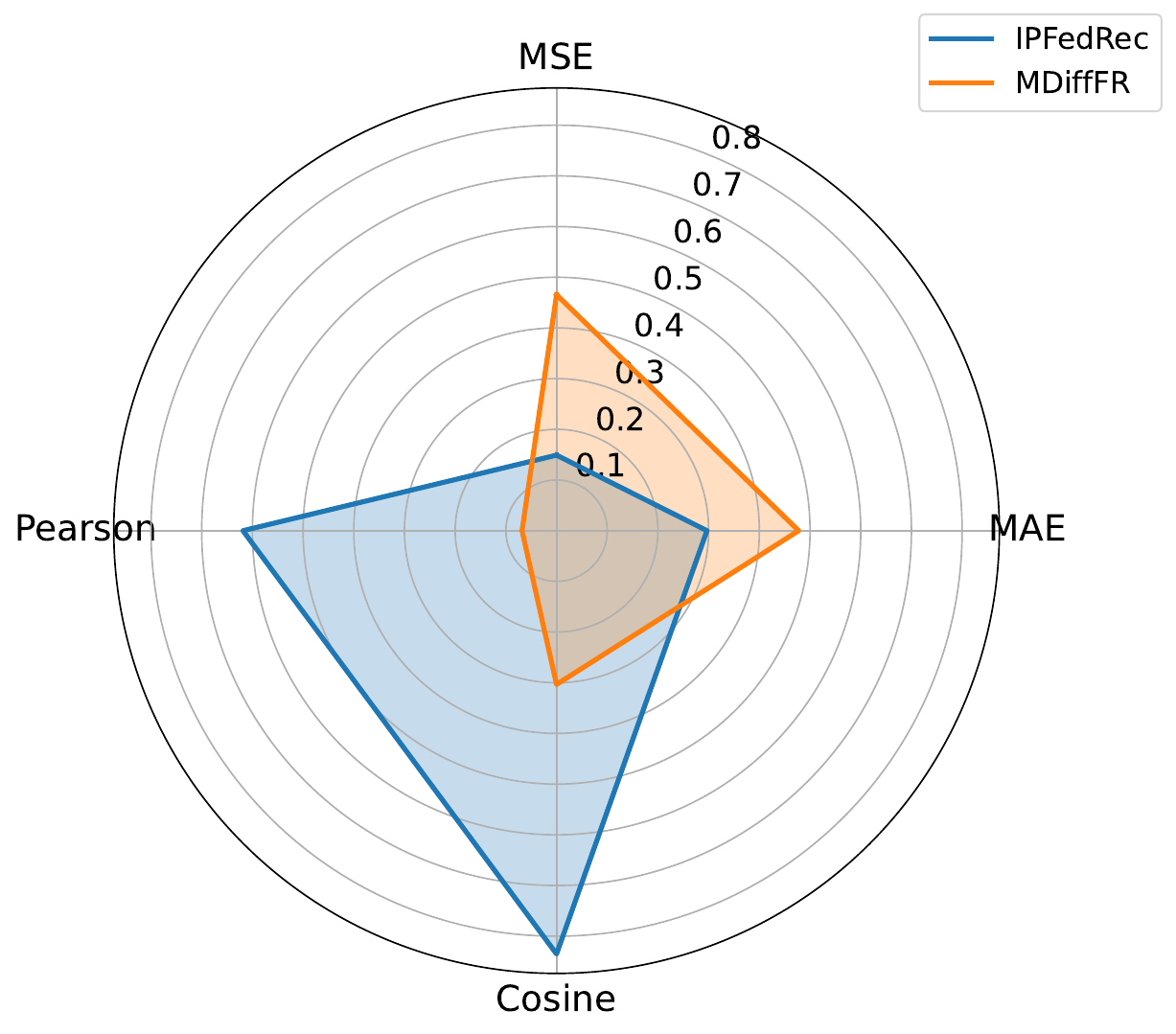}
        \caption{KU}
    \end{subfigure}
    \begin{subfigure}[t]{0.24\columnwidth}
        \includegraphics[width=\linewidth]{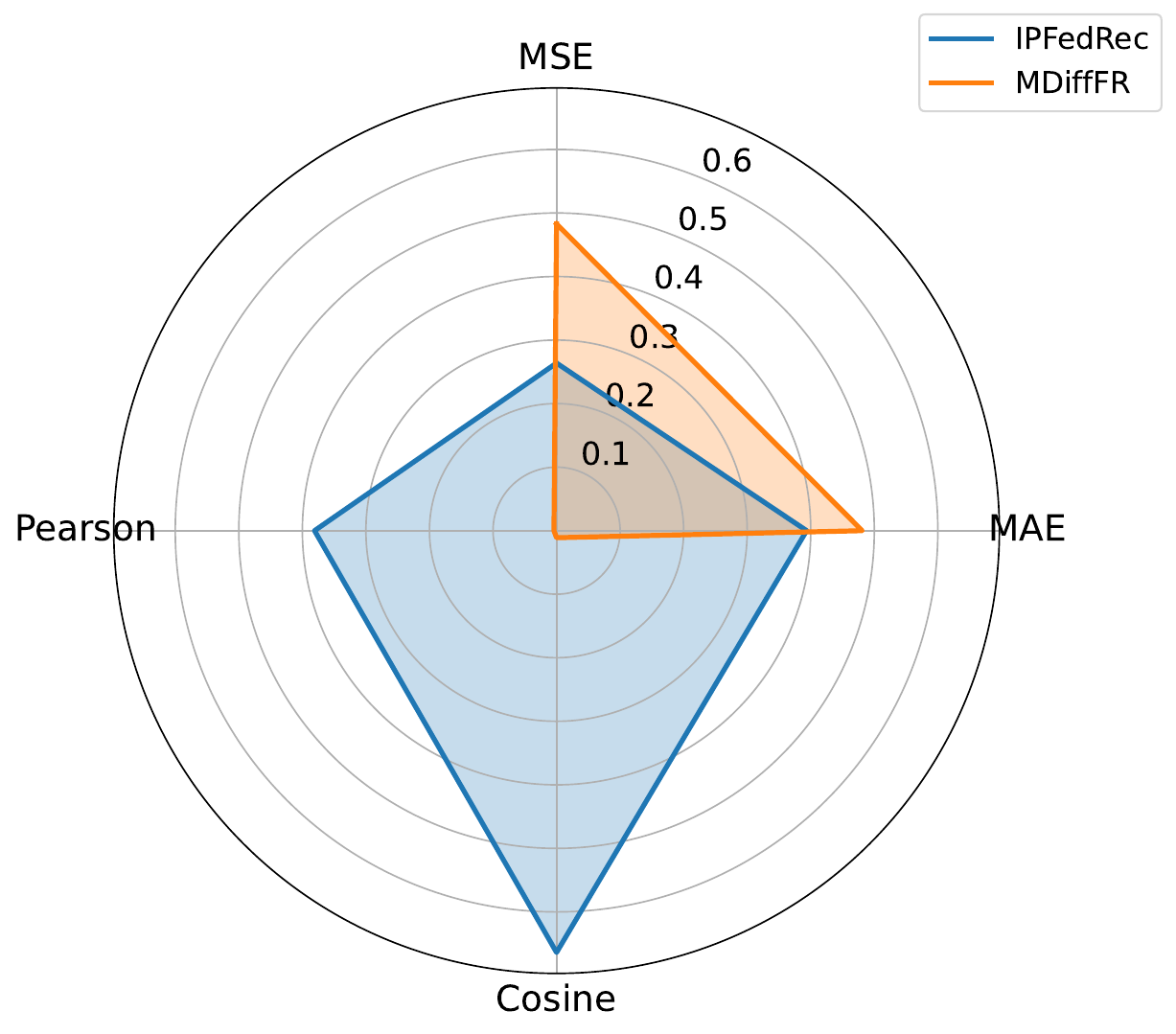}
        \caption{Food}
    \end{subfigure}
    \begin{subfigure}[t]{0.24\columnwidth}
        \includegraphics[width=\linewidth]{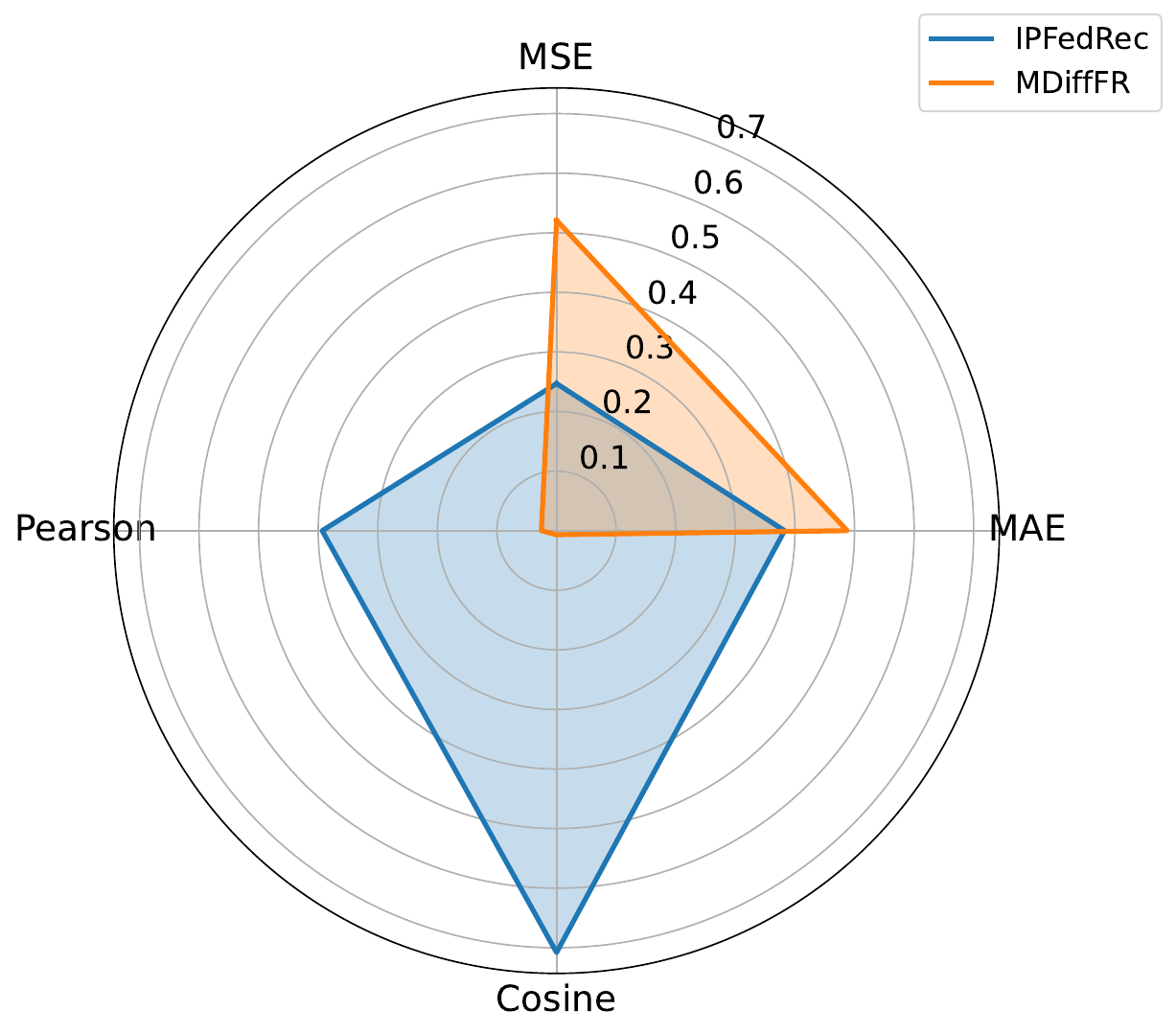}
        \caption{Dance}
    \end{subfigure}
    \begin{subfigure}[t]{0.24\columnwidth}
        \includegraphics[width=\linewidth]{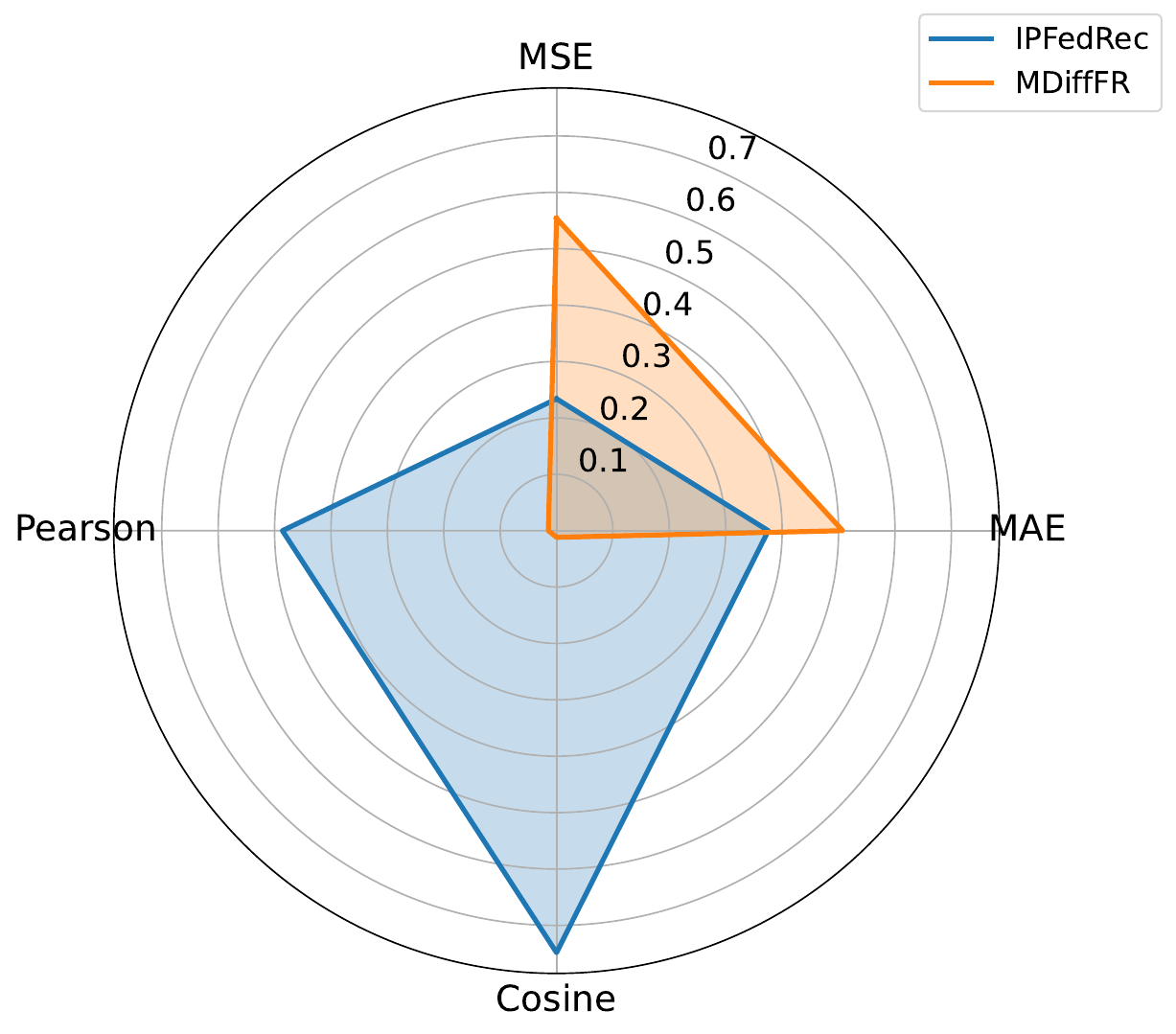}
        \caption{Movie}
    \end{subfigure}
    \end{minipage}
    \caption{Radar chart of different methods’ performance against inversion attacks over four metrics.}
    \label{fig:Radar_Chart}
\end{figure}

To further evaluate the privacy-preserving capability of MDiffFR, we analyze the structural similarity difference. Specifically, we first compute the pairwise similarities among items based on both the original modality embeddings and the reconstructed embeddings. We then measure the difference between the two similarity matrices. A value closer to zero indicates that the attacker can better reconstruct the original similarity distribution of items, implying weaker privacy protection. As shown in Fig. \ref{fig:Structural_similarity}, we randomly visualize the structural similarity differences of 20 items. The differences for IFedRec are much closer to zero, indicating poorer privacy protection. In contrast, MDiffFR exhibits more pronounced deviations, suggesting that the attacker fails to accurately reconstruct the original similarity distribution.

\begin{figure}[h]
    \centering
    \begin{minipage}{0.85\columnwidth} % 必须用 minipage 包住
    \centering
    \begin{subfigure}[t]{0.24\columnwidth}
        \includegraphics[width=\linewidth]{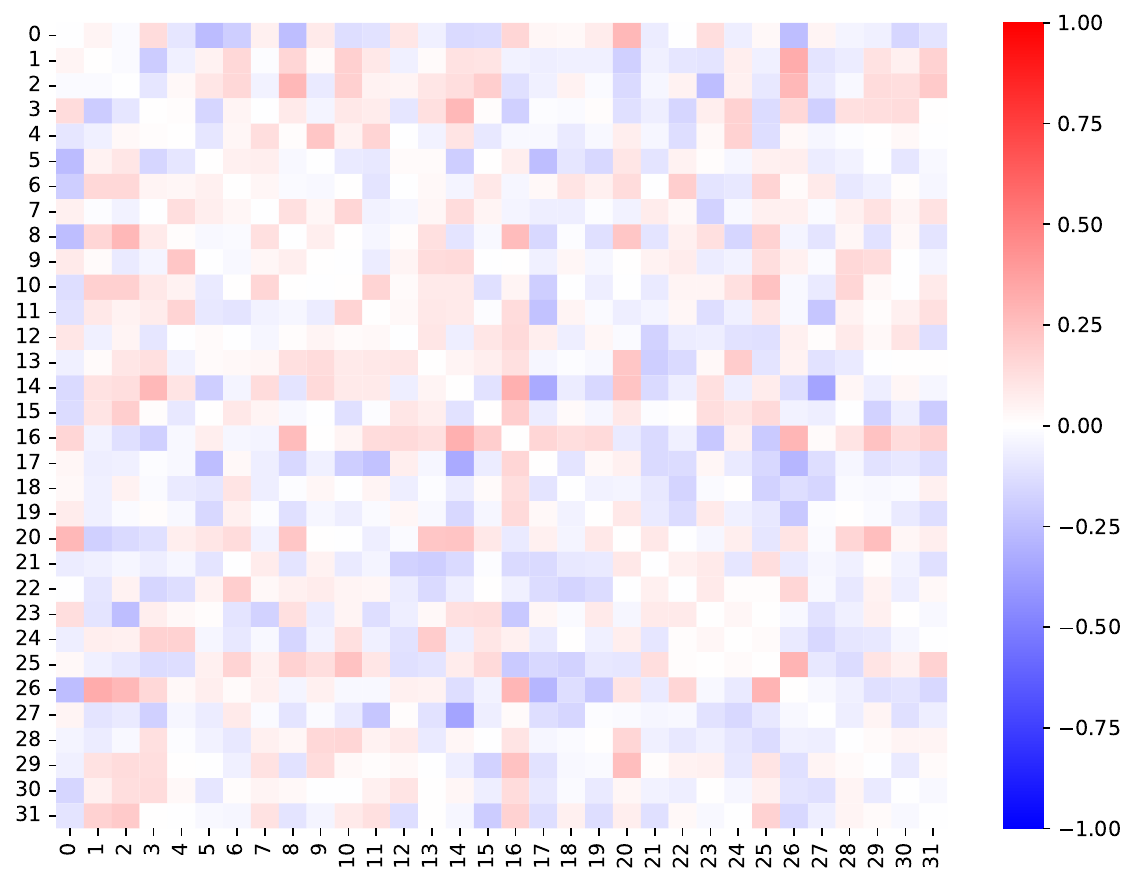}
        \caption{KU-IPFedRec}
    \end{subfigure}
    \begin{subfigure}[t]{0.24\columnwidth}
        \includegraphics[width=\linewidth]{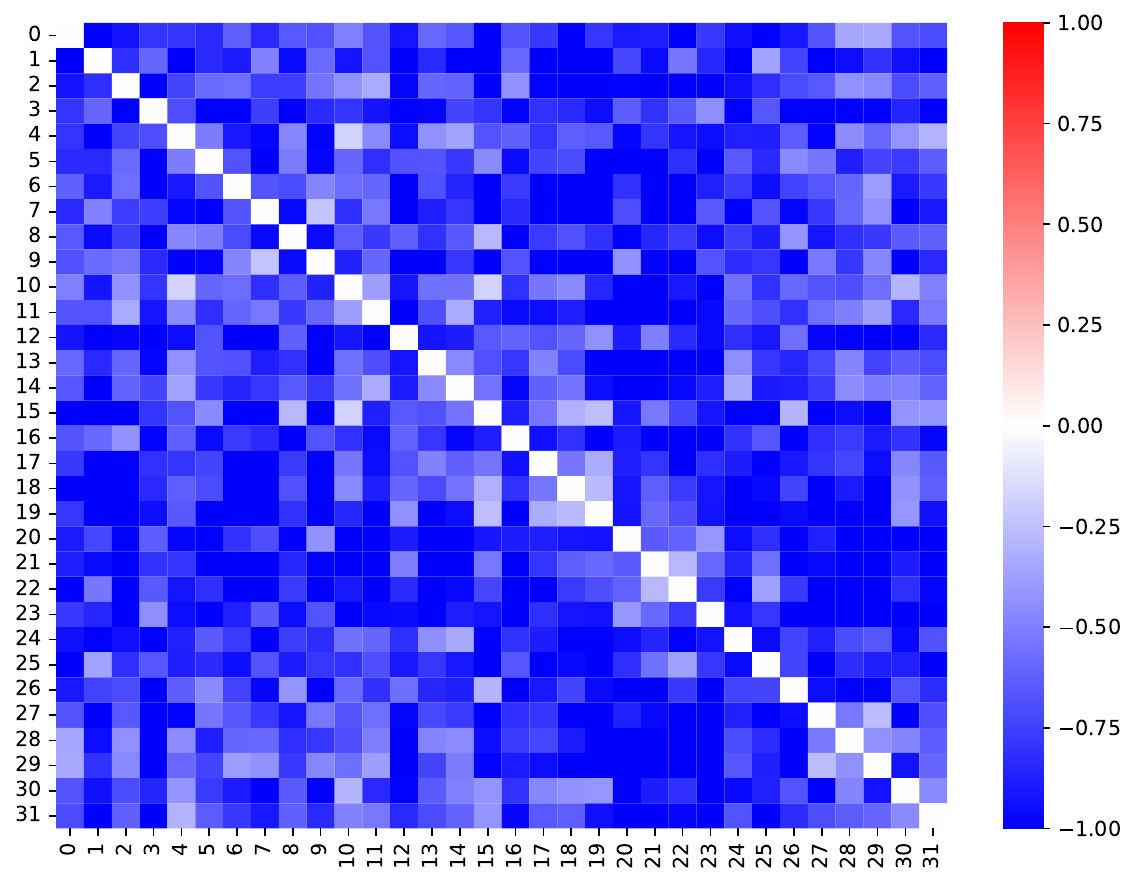}
        \caption{KU-MDiffFR}
    \end{subfigure}
    \begin{subfigure}[t]{0.24\columnwidth}
        \includegraphics[width=\linewidth]{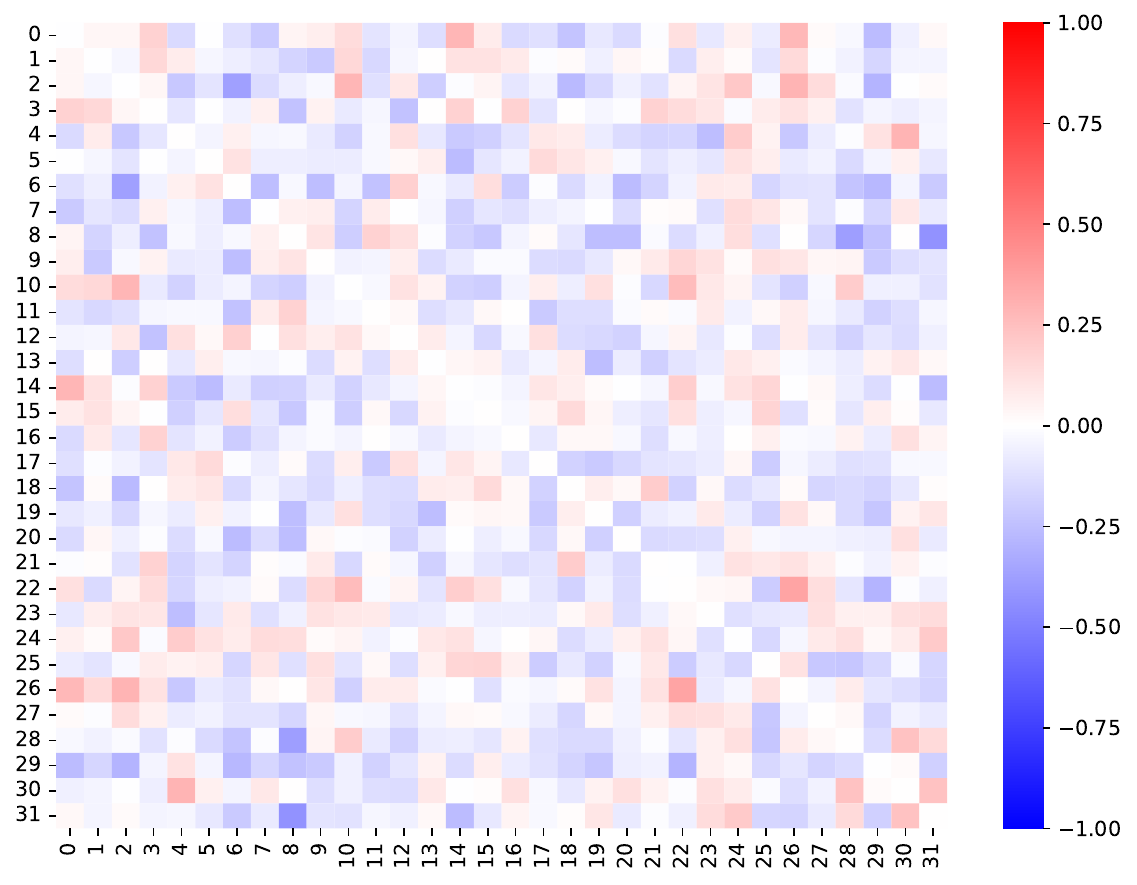}
        \caption{Food-IPFedRec}
    \end{subfigure}
    \begin{subfigure}[t]{0.24\columnwidth}
        \includegraphics[width=\linewidth]{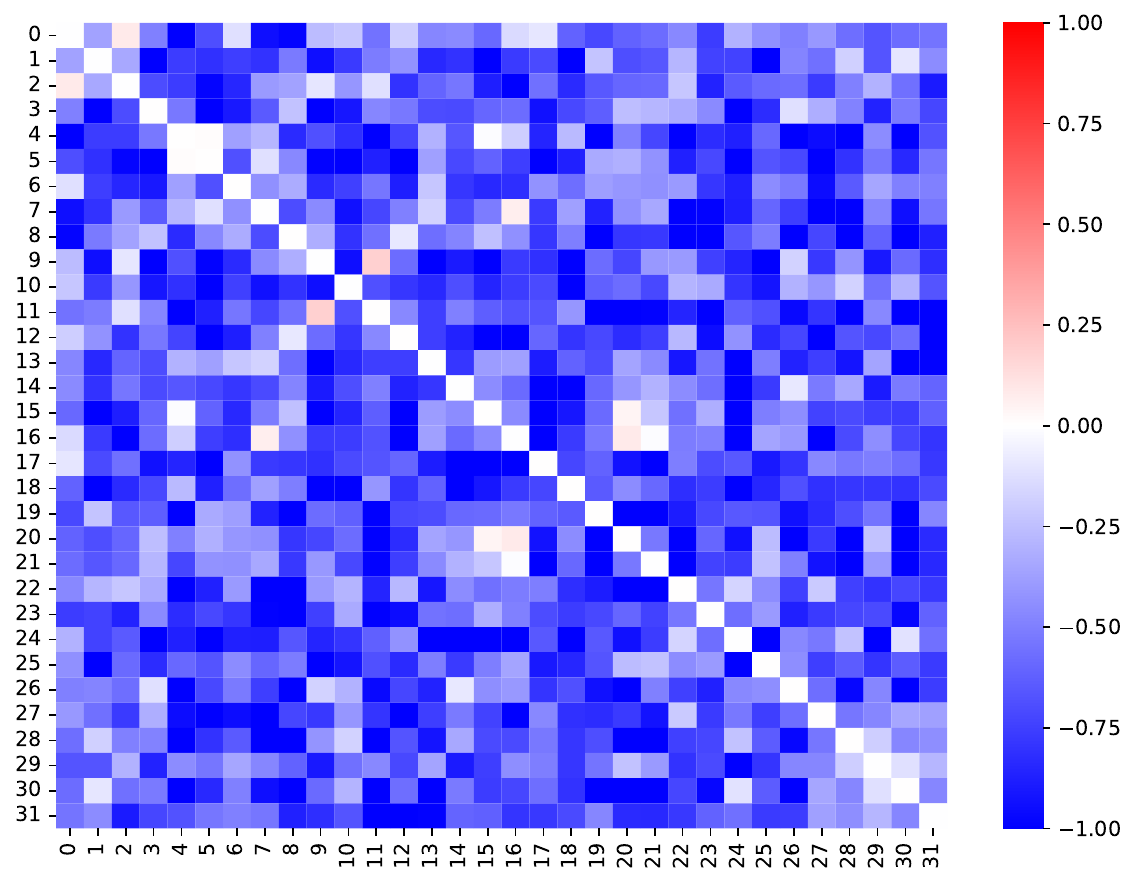}
        \caption{Food-MDiffFR}
    \end{subfigure}
    \\
    \begin{subfigure}[t]{0.24\columnwidth}
        \includegraphics[width=\linewidth]{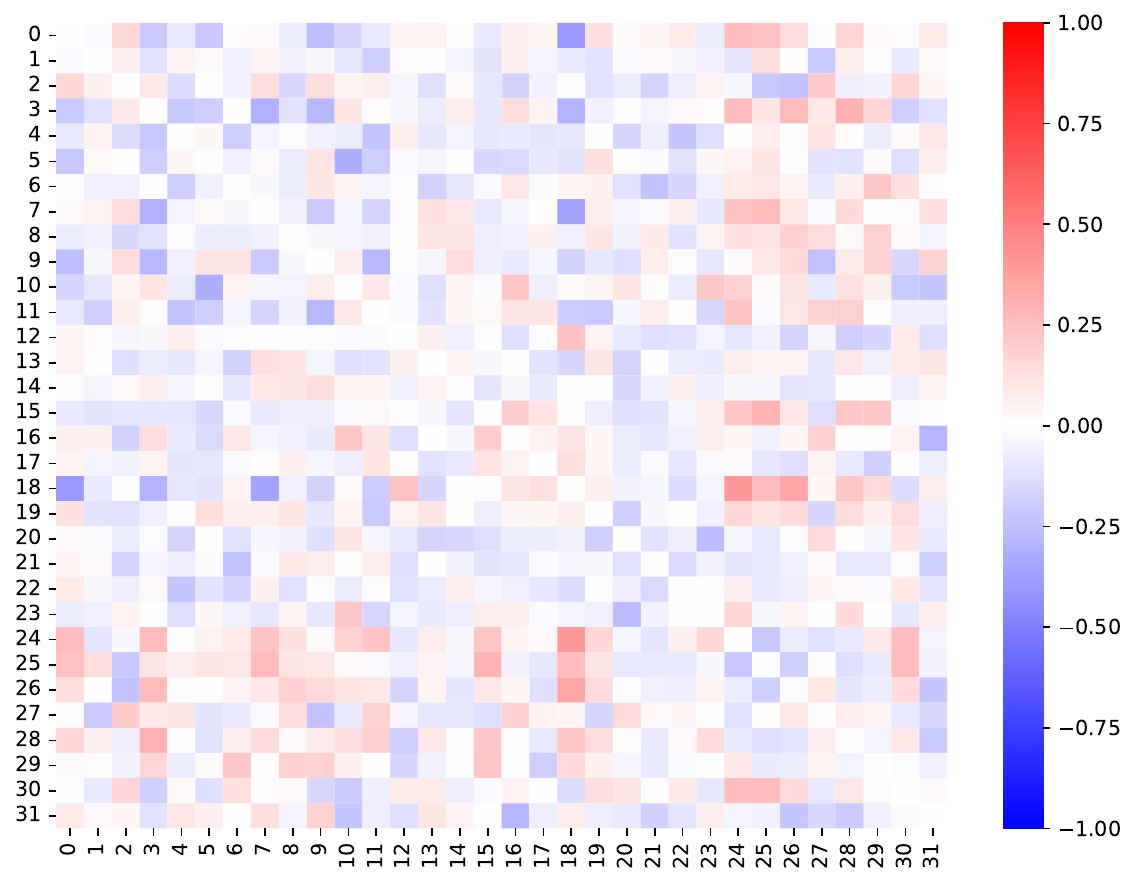}
        \caption{Dance-IPFedRec}
    \end{subfigure}
    \begin{subfigure}[t]{0.24\columnwidth}
        \includegraphics[width=\linewidth]{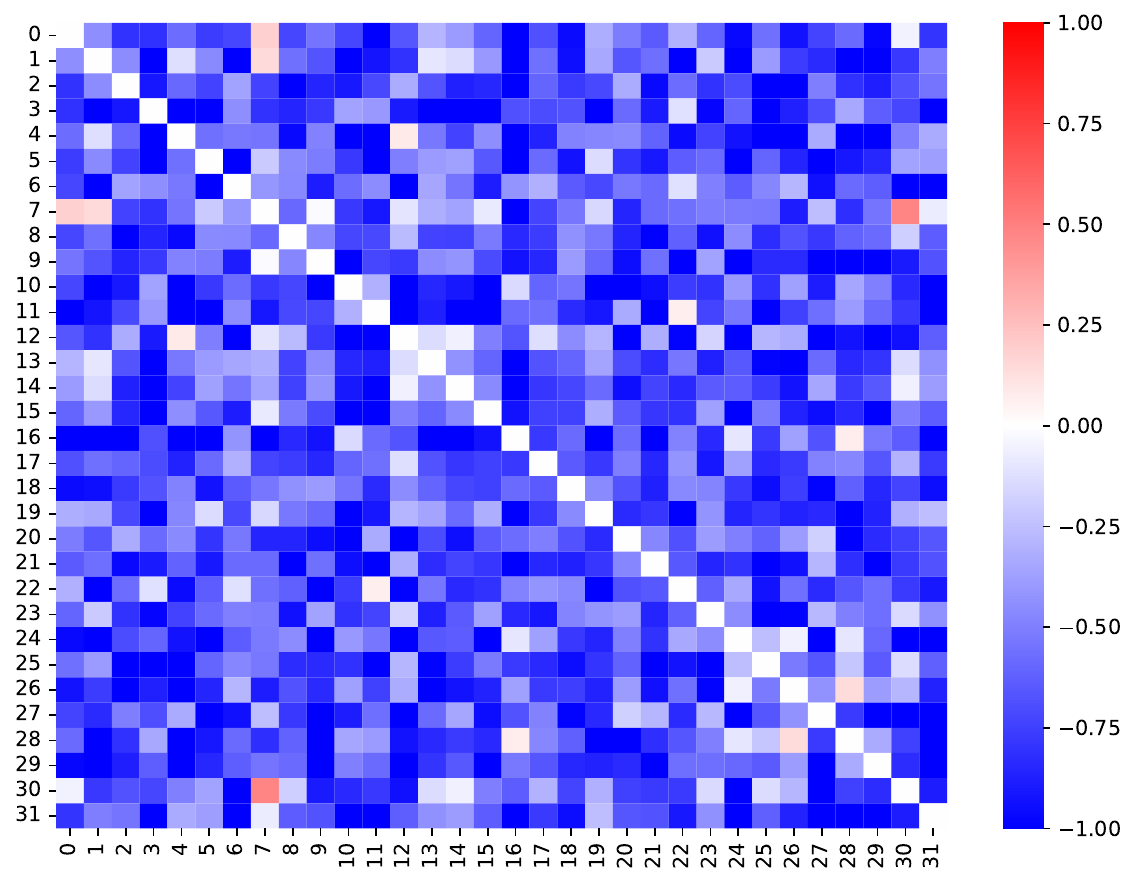}
        \caption{Dance-MDiffFR}
    \end{subfigure}
    \begin{subfigure}[t]{0.24\columnwidth}
        \includegraphics[width=\linewidth]{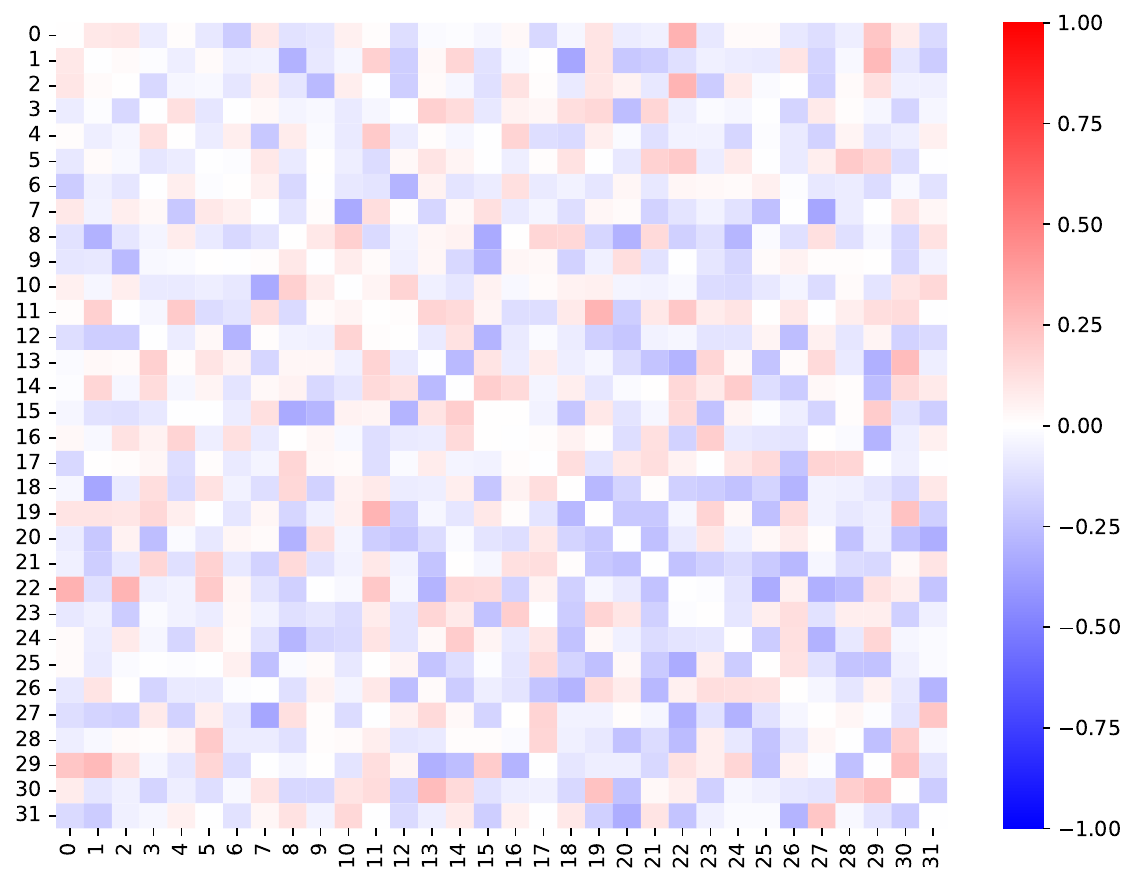}
        \caption{Movie-IPFedRec}
    \end{subfigure}
    \begin{subfigure}[t]{0.24\columnwidth}
        \includegraphics[width=\linewidth]{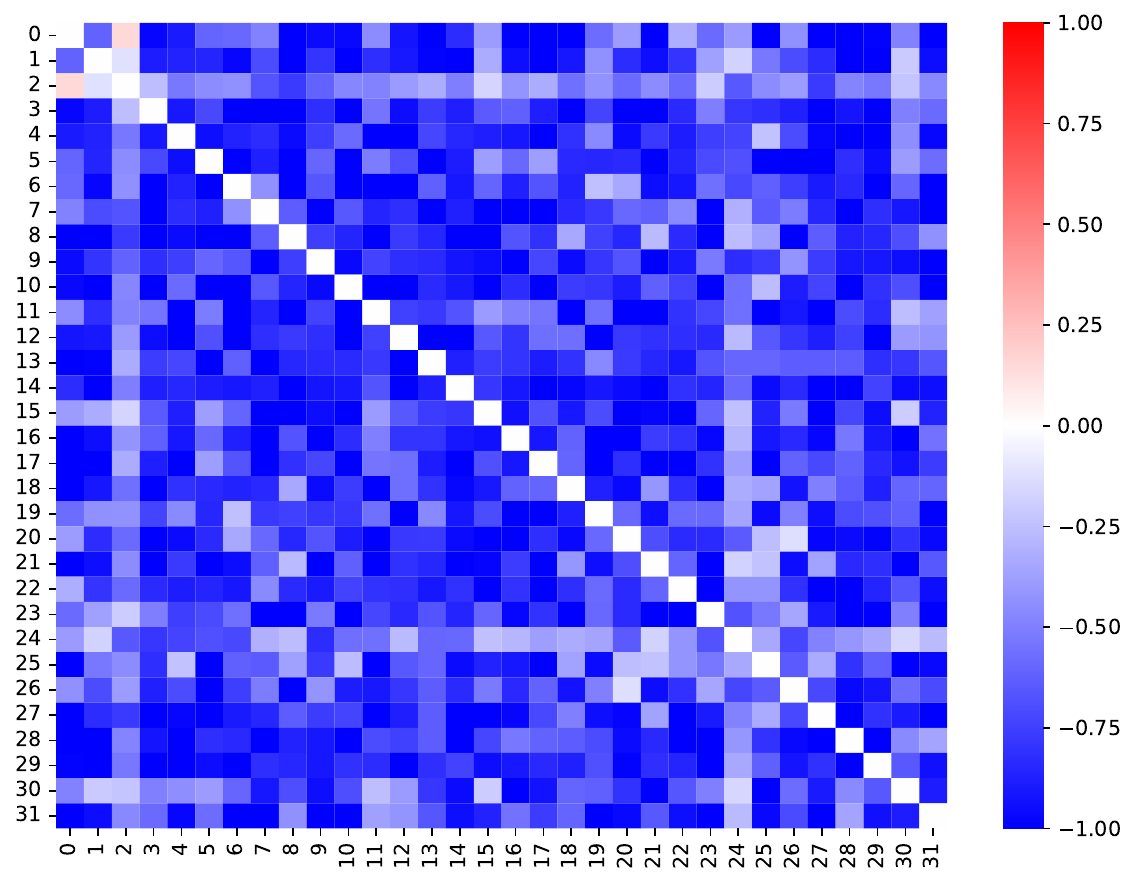}
        \caption{Movie-MDiffFR}
    \end{subfigure}
    \end{minipage}
    \caption{Structural similarity differences of generated embeddings for cold-start items.}
    \label{fig:Structural_similarity}
\end{figure}

In addition, to further enhance the differential privacy guarantees of MDiffFR, we incorporate Local Differential Privacy (LDP) into the framework to perturb the data uploaded to the server. Specifically, we add zero-mean Laplace noise $\epsilon \sim \text{Laplace}\left( 0, \delta \right)$ to the item embeddings $\bm{e}_u$ of client $u $, where $\delta$ denotes the noise intensity. A larger $\delta$ represents stronger perturbations to the embeddings, thereby providing higher privacy protection.
Unlike traditional approaches, our model generates embeddings based on noise sampled from a prior distribution, which makes it more robust to noise. Therefore, we set a higher noise intensity $\delta=[1, 10, 20, 30, 40, 50]$. As shown in Table \ref{tab:delta-food}, the experimental results indicate that as the noise intensity increases, the model's performance experiences only a slight decrease, demonstrating that differential privacy can be maintained without substantially compromising performance.

\begin{table}[!htbp]
\centering
\small
\setlength{\tabcolsep}{4pt}
\renewcommand{\arraystretch}{1.2}
\caption{Impact of the noise intensity $\delta$ for MDiffFR.}
\begin{tabular}{c|ccccccc}
\hline
\multirow{2}{*}{\textbf{Metric}} & \multicolumn{7}{c}{\textbf{Intensity} $\delta$ } \\
\cline{2-8}

{} & \textbf{0} & \textbf{1} & \textbf{10} & \textbf{20} & \textbf{30} & \textbf{40} & \textbf{50} \\
\hline
Recall    & \textbf{19.23} & 18.06 & 15.36 & 16.09 & 15.63 & 18.48 & 17.24 \\
Precision & \textbf{0.62}  & 0.56  & 0.48  & 0.51  & 0.50  & 0.59  & 0.55 \\
NDCG      & \textbf{8.08}  & 6.34  & 5.81  & 5.88  & 5.93  & 7.16  & 6.23 \\
\hline
\end{tabular}
\label{tab:delta-food}
\end{table}

\subsection{Hyper-parameter Analysis (RQ5)}
% \subsubsection{Diffusion activation.}
% During the early stage of training in FR, item embeddings have not yet fully captured local interaction knowledge and global collaborative information. Therefore, we disable the diffusion model in the initial rounds to avoid fitting the undertrained item embeddings. As shown in Fig. \ref{fig:diffusion_activation}, the diffusion model is activated starting from different training rounds. As the diffusion activation round increases, the model performance initially degrades to varying extents. However, when the diffusion model is activated from round 50, the number of diffusion training rounds is effectively reduced by half, while the performance becomes comparable to training with diffusion from the beginning. This suggests that higher-quality initial embeddings enable the diffusion model to more quickly learn data distributions that are more effective for recommendations.

% \subsubsection{Embedding size.}
MDiffFR captures the distribution of item representations to generate embeddings for cold-start items. Therefore, the embedding dimension plays a crucial role in determining its ability to accurately capture such distributions. As shown in Fig. \ref{fig:embedding_size}, we investigate the impact of item embedding dimensionality on model performance. When the embedding dimension is set to 16 or 32, the model exhibits inferior performance across all four datasets, suggesting that excessively low-dimensional embeddings are insufficient to capture the intrinsic distribution of the item representations. The model achieves its optimal or near-optimal performance on all four datasets when the embedding dimension is increased to 64. Although a minor improvement is still observed on the Food and Movie datasets with higher dimensions, further increasing the dimension to 128 leads to a twofold rise in computational and storage costs, while yielding only negligible performance gains. Consequently, such a configuration presents a suboptimal balance between efficiency and effectiveness in practical scenarios. Hence, we adopt 64-dimensional embeddings as the default configuration.

% \begin{figure}[!t]
% \centering
% % \includegraphics[width=\columnwidth]{./image/start_diffusion.pdf}
% \includegraphics[width=0.6\columnwidth]{./image/diffusion_activation.pdf}
% \caption{Effect of diffusion activation.}
% \label{fig:diffusion_activation}
% \end{figure}

% \begin{figure}[h]
% \centering
% \includegraphics[width=\columnwidth]{./image/embedding_size.pdf}
% \caption{Effect of embedding size.}
% \label{fig:embedding_size}
% \end{figure}
\begin{figure}[h]
    \centering
    \begin{subfigure}[t]{0.24\columnwidth}
        \includegraphics[width=\linewidth]{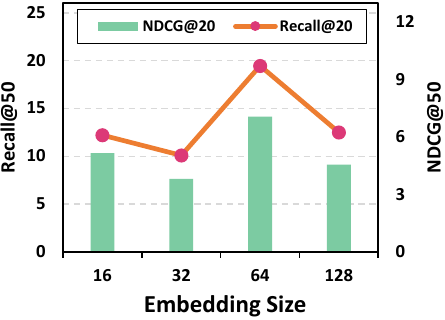}
        \caption{KU}
    \end{subfigure}
    \begin{subfigure}[t]{0.24\columnwidth}
        \includegraphics[width=\linewidth]{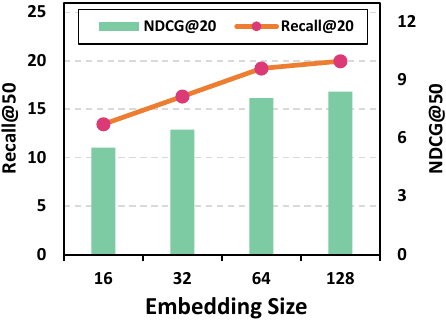}
        \caption{Food}
    \end{subfigure}
    \begin{subfigure}[t]{0.24\columnwidth}
        \includegraphics[width=\linewidth]{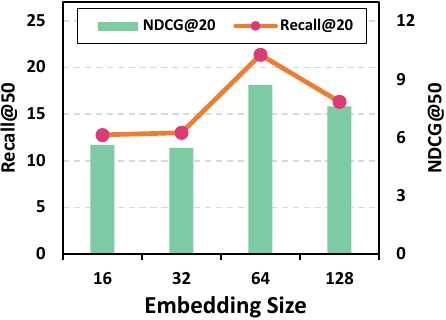}
        \caption{Dance}
    \end{subfigure}
    \begin{subfigure}[t]{0.24\columnwidth}
        \includegraphics[width=\linewidth]{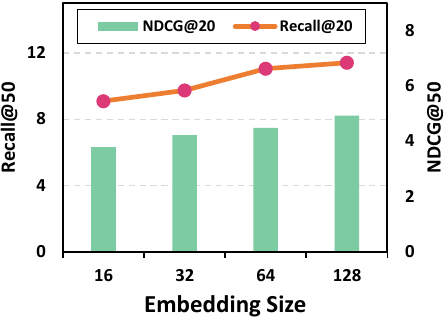}
        \caption{Movie}
    \end{subfigure}
    \caption{Effect of embedding size.}
    \label{fig:embedding_size}
\end{figure}

\section{Conclusion}
In this paper, we empirically find that existing FR methods based on the mapping paradigm tend to cause embedding misalignment when tackling the item cold-start problem, and may pose potential privacy risks due to their reliance on one-to-one mappings of item attributes. To address these issues, we propose MDiffFR, a novel generation-based method that learns the distribution of global item embeddings to generate embeddings for cold-start items. Extensive experiments show that MDiffFR not only outperforms state-of-the-art baselines and alleviates embedding misalignment, but also exhibits stronger privacy protection under adversarial inversion attacks. Meanwhile, our generation-based approach offers a promising new perspective for effectively addressing the item cold-start problem in FRs.

% \begin{acks}
% \end{acks}

\bibliographystyle{ACM-Reference-Format}
\bibliography{MDiffFR}

% \appendix









\end{document}